\documentclass[reqno, eucal]{amsart}

\usepackage[mathscr]{eucal} %  use \EuScript (\mathcal unchanged)

\numberwithin{equation}{section}

\usepackage[dvips]{graphicx}
\usepackage[dvips]{color}
\usepackage{amsmath,amsfonts,amssymb,amsthm,amscd}
\usepackage{epic}
\usepackage{eepic}
\usepackage{longtable}
\usepackage{array}

\newtheorem{theorem}{Theorem}[section]

\newtheorem{proposition}[theorem]{Proposition}

\theoremstyle{definition}

\newtheorem{example}[theorem]{Example}
\newtheorem{remark}[theorem]{Remark}

\newcommand{\Z}{{\mathbb Z}}

\newcommand{\ok}{{\rm{\bf k}}}
\newcommand{\OK}{{\rm{\bf K}}}
\newcommand{\am}{{\rm{\bf a}}^- }
\newcommand{\ap}{{\rm{\bf a}}^+ }
\newcommand{\Am}{{\rm{\bf A}}^- }
\newcommand{\Ap}{{\rm{\bf A}}^+ }
\newcommand{\ichi}{{\bf 1} }

\begin{document}
\title[Tetrahedron and 3D Reflection equations]
{Tetrahedron and 3D reflection equations
from quantized algebra of functions}

\author{Atsuo Kuniba}
\email{atsuo@gokutan.c.u-tokyo.ac.jp}
\address{Institute of Physics, Graduate School of Arts and Sciences,
University of Tokyo, Komaba, Tokyo 153-8902, Japan}

\author{Masato Okado}
\email{okado@sigmath.es.osaka-u.ac.jp}
\address{Department of Mathematical Science,
Graduate School of Engineering Science,
Osaka University, Toyonaka, Osaka 560-8531, Japan}

%\dedicatory{}

\maketitle

\vspace{1cm}
\begin{center}{\bf Abstract}
\end{center}
\vspace{0.4cm}
Soibelman's theory of quantized function algebra $A_q(\mathrm{SL}_n)$ 
provides a 
representation theoretical scheme to construct a solution of the 
Zamolodchikov tetrahedron equation.
We extend this idea originally due to 
Kapranov and Voevodsky to $A_q(\mathrm{Sp}_{2n})$ and 
obtain the intertwiner $K$ corresponding to the quartic Coxeter relation.
Together with the previously known 3-dimensional (3D) $R$ matrix, 
the $K$ yields the first ever solution to the 3D analogue of the reflection equation
proposed by Isaev and Kulish.  
It is shown that matrix elements of $R$ and $K$ are polynomials in $q$
and that there are combinatorial and birational counterparts for $R$ and $K$.
The combinatorial ones arise either at $q=0$ 
or by tropicalization of the birational ones.
A conjectural description for the type $B$ and $F_4$ cases 
is also given.

\section{Introduction}\label{sec:intro}

The tetrahedron equation \cite{Zam80,Zam81} is the 3-dimensional (3D) analogue 
of the Yang-Baxter equation \cite{Bax}.
It has been well recognized by now that 
quantum groups \cite{D86, J1} provide a comprehensive 
framework for the algebraic aspect of the latter.
The tetrahedron equation is more challenging but 
many efforts and results continue to emerge until today.  
See \cite{Bax86, Korep89, MN89, BB92, LM93, KMS93, 
Korep93, H94, KV, Korep95, SW95, Ka96, SMS96,
S98, KKS98, BS, BMS,KuS} for example. 
In this paper we study the tetrahedron equation 
and its generalizations based on 
representation theory of quantized algebra of functions.
Let us briefly review this approach which is seemingly long forgotten, 
to motivate our work.

Let $\mathfrak{g}$ be a classical simple Lie algebra and 
$G$ be the corresponding Lie group.
The quantized algebra of functions on $G$ 
or $q$ deformation of the coordinate ring of $G$
is the Hopf algebra dual to the quantized universal enveloping algebra 
$U_q(\mathfrak{g})$.
We denote it by $A_q(G)$ in this paper
and assume that $q$ is generic unless otherwise stated.
It has been studied from a variety of viewpoints.
See \cite{D86,RTF,NYM, VS,So1,So2} for example.
The simplest one is 
$A_q(\mathrm{SL}_2) = \langle t_{11}, t_{12}, t_{21}, t_{22}\rangle$ 
with the relations 
\begin{equation}\label{sl2}
\begin{split}
&t_{11}t_{21} = qt_{21}t_{11},\quad
t_{12}t_{22} = qt_{22}t_{12},\quad
t_{11}t_{12} = qt_{12}t_{11},\quad
t_{21}t_{22} = qt_{22}t_{21},\\
&[t_{12},t_{21}]=0,\quad [t_{11},t_{22}]=(q-q^{-1})t_{21}t_{12},
\quad t_{11}t_{22}-qt_{12}t_{21}=1.
\end{split}
\end{equation}
It has the irreducible representation in terms of 
the $q$-oscillator acting on the Fock space \cite{VS}.
See (\ref{qoa})--(\ref{repA}).
Irreducible representations of $A_q(G)$ for 
general $G$ were classified by Soibelman \cite{So1,So2}.
Associated with each vertex $i$ of the Dynkin diagram,
$A_q(G)$ has an irreducible representation $\pi_i$ 
which factors through the projection to   
$A_q(\mathrm{SL}_2)$ corresponding to $i$. 
The representations 
$\pi_1, \ldots, \pi_n\, (n = \mathrm{rank}\, G)$ 
play the role of fundamental representations. 
General irreducible representations are in one to one 
correspondence with elements of the Weyl group 
$W(G)=\langle s_1, \ldots, s_n \rangle$
up to some ``torus degrees of freedom".
More concretely if 
$w = s_{i_1}\cdots s_{i_r}\in W(G)$ is a reduced expression 
by the simple reflections, 
the corresponding irreducible $A_q(G)$ module is realized as the 
tensor product of the fundamental ones as
$\pi_{i_1} \otimes \cdots \otimes \pi_{i_r}$.
A crucial consequence of this claim is the equivalence 
$\pi_{i_1} \otimes \cdots \otimes \pi_{i_r}
\simeq 
\pi_{j_1} \otimes \cdots \otimes \pi_{j_r}$
for arbitrary reduced expressions 
$w = s_{j_1}\cdots s_{j_r}$.
In particular it ensures the existence of a unique (up to normalization) isomorphism between such tensor products,  
which we call  the {\em intertwiner}.

In \cite{KV} Kapranov and Voevodsky 
found an application of these results for type $A$ 
to the tetrahedron equation.
The Coxeter relations $s_2s_1s_2 = s_1s_2s_1$ 
and $s_3s_2s_3 = s_2s_3s_2$ in $W(\mathrm{SL}_4)$ 
imply the equivalence of the irreducible
$A_q(\mathrm{SL}_4)$ modules 
$\pi_2 \otimes \pi_1\otimes \pi_2
\simeq \pi_1 \otimes \pi_2 \otimes \pi_1$ and 
$\pi_3 \otimes \pi_2\otimes \pi_3
\simeq \pi_2 \otimes \pi_3 \otimes \pi_2$,
therefore the existence of unique (up to normalization) intertwiners 
$\Phi^{(i)}$ satisfying 
\begin{equation}\label{phi12}
\begin{split}
(\pi_2 \otimes \pi_1 \otimes \pi_2) \circ \Phi^{(1)} &= 
\Phi^{(1)} \circ(\pi_1 \otimes \pi_2 \otimes \pi_1),\\
(\pi_3 \otimes \pi_2 \otimes \pi_3) \circ \Phi^{(2)} &= 
\Phi^{(2)} \circ(\pi_2 \otimes \pi_3 \otimes \pi_2).
\end{split}
\end{equation}
On the other hand the two reduced expressions of the longest element
$s_1s_2s_3s_1s_2s_1 = s_3s_2s_3s_1s_2s_3$
lead to the equivalence
$\pi_1\otimes \pi_2\otimes \pi_3\otimes \pi_1\otimes \pi_2\otimes \pi_1
\simeq 
\pi_3\otimes \pi_2\otimes \pi_3\otimes \pi_1\otimes \pi_2\otimes \pi_3$.
Up to transposition of components, the intertwiner for this can be constructed 
as the composition of  the form $\Phi^{(2)} \Phi^{(1)} \Phi^{(2)} \Phi^{(1)} $
in two ways, which parallel the transformations of the reduced expressions
using the Coxeter relations. 
See (\ref{2ways}) for the precise description.
Since the intertwiner is unique, it enforces the consistency condition
of the form
$\Phi^{(2)} \Phi^{(1)} \Phi^{(2)} \Phi^{(1)}  = 
\Phi^{(2)} \Phi^{(1)} \Phi^{(2)} \Phi^{(1)} $ up to transposition operators
(cf. (\ref{phpA})). 
It turns out that $\Phi^{(1)}$ and $\Phi^{(2)}$ yield essentially the same 
matrix acting on the tensor cube of the Fock space.
Regarding them as the 3D $R$ matrix, 
the consistency condition is nothing but the tetrahedron equation.

The discovery by Kapranov and Voevodsky was  
rephrased in \cite{KS}, but 
did not seem to have plentiful citations in the mathematical physics 
community working on the tetrahedron equation, possibly due to some unfortunate misprints.
However in retrospect, 
it was offering a proper quantum group theoretical framework 
for the ideas being developed independently at about the same time 
in several guises called 
local Yang-Baxter equation, 
tetrahedral Zamolodchikov algebra and 
quantum Korepanov equation, etc   
\cite{MN89,Korep93, Korep95,KKS98}.
In particular the 3D $R$ matrix 
obtained relatively recently by
Bazhanov, Sergeev and Mangazeev \cite{BS,BMS} 
can be identified with the intertwiner 
that follows from the Kapranov-Voevodsky approach.
See (\ref{pp}).
The essence of the strategies in the works
\cite{MN89,Korep93, Korep95,KKS98,BS,BMS} may roughly be stated as 
upgrading the {\em equality} in the Yang-Baxter relation to 
a {\em transformation}  $\varphi$ of some triple product
\begin{equation}\label{rsss}
S_1S_2S_1 = S_2S_1S_2 \quad \Longrightarrow\quad 
\varphi: S_1S_2S_1 \longmapsto S'_2S'_1S'_2
\end{equation}
and to produce a solution of the tetrahedron equation from $\varphi$.
A similar idea was considered 
as ``vectorization of triangle equations" in \cite[sec.6.9]{KV}.
The intertwining relation (\ref{phi12}) is a realization of it  
as $\varphi = \mathrm{Ad}(\Phi^{(i)})$.
From this viewpoint the prescription (\ref{rsss}) is traced back to the 
transformations of reduced expressions of
the Weyl group elements by means of 
the cubic Coxeter relation.
The 3D $R$ matrix and the conventional 2D quantum $R$ matrices
are thus put on a parallel footing.  
They are pinned as the intertwiners for $A_q$ modules 
and $U_q$ modules, respectively.
The representations of $A_q$ do not have the 
variety like the highest weight modules for $U_q$.
However their intertwiners are known to generate a broad class of 
2D quantum $R$ matrices in some cases \cite{BS, KuS}.

\begin{table}[h]
\vspace{0.2cm}
\begin{tabular}{c|c}
Weyl group & Factorized scattering in 3D\\
\hline \vspace{-0.3cm}\\
reduced expressions\; & \; multi-string states\\
cubic Coxeter relation \;
& \;3D $R$ (scattering amplitude)\\
the longest element \;
& \; tetrahedron equation
\end{tabular}
\end{table}

The story in type $A$
digested so far naturally motivates us to explore 
the general $G$ case.
It may also be viewed as the 3D analogue of Cherednik's 
generalization of the Yang-Baxter equation 
along the classical Coxeter systems \cite{Ch84}.
Our aim in this paper is to launch such a result 
for those $G$ having a double arrow in the Dynkin diagrams
mainly along type $C$, i.e. $G=\mathrm{Sp}_{2n}$.
We introduce the quantized algebra of functions $A_q(\mathrm{Sp}_{2n})$ 
following \cite{RTF}. 
The $n=3$ case already covers the generic situation on which we shall focus.
Irreducible representations $\pi_i (i=1,2,3)$ are presented 
in terms of the $q_i$-deformed oscillators, where
$(q_1,q_2,q_3) = (q,q,q^2)$  
reflecting the squared lengths of the three simple roots of $\mathrm{Sp}_6$.
The Coxeter relations in the 
Weyl group $W(\mathrm{Sp}_6)=\langle s_1, s_2, s_3\rangle$ include 
$s_1s_2s_1 = s_2s_1s_2$ and $s_2s_3s_2s_3 = s_3s_2s_3s_2$.
Accordingly we have the equivalence of the tensor product representations
\begin{equation}\label{cri}
\pi_1 \otimes \pi_2 \otimes\pi_1
\simeq \pi_2 \otimes \pi_1 \otimes\pi_2,\qquad
\pi_2 \otimes \pi_3 \otimes \pi_2 \otimes\pi_3 
\simeq \pi_3 \otimes \pi_2 \otimes \pi_3 \otimes\pi_2.
\end{equation}
The problem is to construct their intertwiners explicitly
and formulate a generalization of the tetrahedron equation 
that should emerge as a consistency condition 
among them.
We find that the former in (\ref{cri}) leads to essentially the
same intertwiner $R$ (or $\Phi$ as in (\ref{rphA13})) as type $A$.
The intertwiner for the latter is a new object, which will be denoted by 
$K$ (or $\Psi$ as in (\ref{pskp})).
We present its explicit formula in Theorem \ref{th:main}.
The $R$ and $K$ act on
tensor cube and tensor quartet of the Fock space, respectively.
However they are locally finite with respect to the natural basis of the 
Fock space, i.e. they are decomposed into 
direct sums of finite dimensional matrices 
specified by conserved quantities.
The consistency condition involving 
$R$ and $K$ is derived from the 
longest element of $W(\mathrm{Sp}_6)$.
The result takes the form $RRKRRKK = KKRRKRR$.
See (\ref{rkeq1}) and (\ref{rkeq2}).
We call it the 3D reflection equation as it is a
natural 3D analogue of the 2D reflection equation 
going back to \cite{Ch84, Sk88, KSk92}.
The physical meaning of it is a factorization condition of 
3 strings scattering in 3D with boundary reflections, where 
the $R$ and $K$ stand for the amplitudes 
of 3 string scattering and reflection at the boundary respectively.
Such a 3D system was originally considered by Isaev and Kulish \cite{IK}
who proposed the ``tetrahedron reflection equation".
It turns out that our 3D reflection equation coincide 
exactly with the constant version (spectral parameter-free case)
of their tetrahedron reflection equation\footnote{
The authors are indebted to A. Isaev for a comment on this point.}.
Our Theorem \ref{th:main} yields the first ever solution to it.

There is another persistent theme in this paper.
We have seen that the $R$ and $K$ are intertwiners 
corresponding to the Coxeter relations
$(s_is_j)^{m_{ij}}=1$ with $m_{ij} = 3$ and $m_{ij}=4$, respectively.
They are {\em quantum} objects controlled by the quantum algebra $A_q(G)$.
We shall establish that the both $R$ and $K$ possess 
{\em polynomial} matrix elements in $q$ and form 
a triad with their {\em birational} and {\em combinatorial} counterparts.
They all satisfy tetrahedron and 3D reflection equations in the respective setting. 
The combinatorial ones are bijections among finite sets,
which arise either at $q=0$ of the quantum $R$ and $K$
or by  the {\em  tropicalization (ultradiscretization)} 
of the birational $R$ and $K$.
These features are quite analogous to the quantum $R$ matrices in 2D 
(cf. \cite{KO,IKT}) 
and have been summarized in Table 1.
The birational $R$ and $K$ are 
the maps (\ref{birR}) and (\ref{abcdt}) characterized by 
the identities (\ref{GGG}), (\ref{Xeq}) and (\ref{Yeq}) 
among the generators $G_i(x), X_i(x), Y_i(x)$ 
of the unipotent subgroup of $G$.
These matrix identities are related to the independence of Schubert
cells $X_w$ for Weyl group elements $w$ on the reduced expressions.
Soibelman's theory on $A_q(G)$ forms a quantum analogue of 
this aspect as stressed in \cite{So1,So2}.
We note that the combinatorial 3D $R$ has effectively appeared already in \cite{L}.
The birational 3D $R$ has also been encountered in various contexts.
See for example \cite{BFZ,YY} and \cite{S98,KKS98},
where in the latter reference 
the tetrahedron equation for them has been called functional 
tetrahedron equation.

\begin{table}[h]
\vspace{0.2cm}
\begin{tabular}{c|ccc}
 & quantum & combinatorial & birational\\ 
 \hline\vspace{-0.2cm} \\
 3D $R$ & $\mathscr{R}$ (\ref{rA2}) & $\mathcal{R}$ 
 (\ref{rzero}) & ${\bf R}$ (\ref{birR}) \\
 3D $K$ & $\mathscr{K}$ (Th.\ref{th:main}) & $\mathcal{K}$ 
 (Th.\ref{th:combk}) & ${\bf K}$ (\ref{abcdt}) 
\end{tabular}
\vspace{0.2cm}
 \caption{Triad of 3D $R$ and 3D $K$}
\end{table}

\vspace{0.2cm}
The layout of the paper is as follows.
In Section \ref{sec:A} we begin by reviewing the 
type $A$ case following the idea of the paper \cite{KV} in 1994.
See also \cite{KS}.
It may be viewed as a formulation of many 
parallel ideas and results in \cite{Korep93, Korep95,KKS98,BS,BMS}
in terms of the representation theory of 
$A_q(\mathrm{SL}_n)$ \cite{VS,So1,So2}.
We point out in (\ref{pp})  the coincidence of the 
intertwiner in \cite{KV} (up to misprints) and the
3D $R$ in \cite{BS,BMS}.
Another observation which appears new is 
the polynomiality of the 3D $R$ in $q$ (Remark \ref{re:Rq0}),
which leads to the combinatorial 3D $R$.
The triad of quantum, birational and combinatorial $R$'s are 
formulated in a unified perspective and summarized in Table~1. 

In Section \ref{sec:C} we deal with type $C$ case.
We obtain the intertwiner $K$ corresponding to the quartic Coxeter relation
and show that it also forms the triad with the birational and combinatorial
counterparts.
We formulate the 3D reflection equation involving $R$ and $K$ 
and argue the relation with the physical setting of the 
tetrahedron reflection equation by Isaev and Kulish \cite{IK}.

In Section \ref{sec:BF} 
we present the intertwiners and their consistency conditions 
for type $B$ and $F_4$ cases on conjectural basis. 
They are natural candidates indicated from the results on type $C$. 
The type $B$ case yields the second (conjectural) solution 
of the same 3D reflection equation.

Appendix \ref{app:IR} gives the list of 
intertwining relations for $K$. 
Appendices \ref{app:K} and \ref{app:kq0} contain the 
technical details of the proofs of 
Theorem \ref{th:main} and Theorem \ref{th:combk}.

\section{$\mathrm{SL}$ case}\label{sec:A}

\subsection{\mathversion{bold}Quantized algebra of functions 
$A_q(\mathrm{SL}_n)$}
We begin by recalling the quantized algebra of functions of type $A$.
It has been studied and denoted in many ways as
$\mathrm{Fun(SL}_q(n))$ \cite{RTF}, 
$\mathbb{C}[G]_h$ \cite{VS},
$A(SL_q(n;\mathbb{C}))$ \cite{NYM},
$\mathbb{C}[SU(2)]_q$ ($n=2$ case) \cite{So2},
$\mathbb{C}[GL(n)_q]$ ($GL$ case rather than $SL$) \cite{KV},
$\mathbb{C}[SL_2(\mathbb{C})](q)$ \cite{KS}
and so on.
In this paper we write it as $A_q(\mathrm{SL}_n)$.
The $A_q(\mathrm{SL}_n)$ is a Hopf algebra \cite{Abe} 
generated by $T = (t_{ij})_{1\le i,j\le n}$ with relations.
They are presented in the 
so called $RTT = TTR$ form and the unit quantum determinant condition:
\begin{align}
&\sum_{m,p}R_{ij,mp}t_{mk}t_{pl} 
= \sum_{m,p}t_{jp}t_{im}R_{mp,kl},\label{re1}\\
&\sum_{\sigma \in \mathfrak{S}_n}(-q)^{l(\sigma)}
t_{1\sigma_1}\cdots t_{n\sigma_n} = 1.\label{re12}
\end{align}
Here $l(\sigma)$ denotes the length of the permutation $\sigma$ and 
the structure constant is specified by
\begin{align*}
\sum_{i,j,k,l}R_{ij,kl}E_{ik}\otimes E_{jl}=
q\sum_i E_{ii}\otimes E_{ii} + \sum_{i\neq j}E_{ii}\otimes E_{jj}
+(q-q^{-1})\sum_{i>j}E_{ij}\otimes E_{ji},
\end{align*}
where the indices are summed over $1,2,\ldots, n$, and $E_{ij}$ is a matrix unit.
This matrix is extracted as 
\begin{align*}
&\sum_{i,j,m,l}R_{ij,ml}E_{im}\otimes E_{jl}
= q\lim_{x\rightarrow \infty}x^{-1}R(x)|_{k=q^{-1}}
\end{align*}
from the quantum $R$ matrix $R(x)$ \cite{Baz,J2} for 
the vector representation of $U_q(A^{(1)}_{n-1})$ 
given in \cite[eq.(3.5)]{J2}. 
Explicitly, the relation (\ref{re1}) reads as (see for example \cite{NYM})
\begin{equation*}
\begin{split}
&[t_{ik}, t_{jl}]=\begin{cases}0 &(i<j, k>l),\\
(q-q^{-1})t_{jk}t_{il} & (i<j, k<l),
\end{cases}\\
&t_{ik}t_{jk} = q t_{jk}t_{ik}\; (i<j),\quad 
t_{ki}t_{kj} = q t_{kj}t_{ki}\; (i<j).
\end{split}
\end{equation*} 
The case $n=2$ is given by (\ref{sl2}).
The coproduct is the standard one: 
\begin{equation}\label{copro}
\Delta(t_{ij}) = \sum_k t_{ik}\otimes t_{kj}.
\end{equation}
We will use the same symbol $\Delta$ flexibly to also mean
the multiple coproducts like
$(\Delta\otimes 1) \circ \Delta = (1\otimes \Delta) \circ \Delta$, etc.
We omit the antipode and the counit for they will not be used in this paper.

\subsection{\mathversion{bold}Representations $\pi_i$}\label{subsec:piA}
Let $\mathrm{Osc}_q = 
\langle {{\bf 1}, \rm{\bf a}}^+, {\rm{\bf a}}^-, {\rm{\bf k}} \rangle$ 
be the $q$-oscillator algebra, i.e. 
an associative algebra with the relations
\begin{equation}\label{qoa}
\begin{split}
&{\rm{\bf k}} \,{\rm{\bf a}}^+ = q\,{\rm{\bf a}}^+{\rm{\bf k}},\quad
{\rm{\bf k}}\,{\rm{\bf a}}^- = q^{-1}{\rm{\bf a}}^-{\rm{\bf k}},\\
&{\rm{\bf a}}^- {\rm{\bf a}}^+ = {\bf 1}-q^2{\rm{\bf k}}^2,\quad 
{\rm{\bf a}}^+{\rm{\bf a}}^- = {\bf 1}-{\rm{\bf k}}^2,\quad
[{\bf 1}, {\rm everything}]=0.
\end{split}
\end{equation}
It has a representation on 
the Fock space $\mathcal{F}_q = \oplus_{m\ge 0}{\mathbb C}(q)|m\rangle$:
\begin{align}\label{akf}
{\bf 1}|m\rangle = |m\rangle,\;
{\rm{\bf k}}|m\rangle = q^m |m\rangle,\;
{\rm{\bf a}}^+|m\rangle = |m+1\rangle,\;
{\rm{\bf a}}^-|m\rangle = (1-q^{2m})|m-1\rangle.
\end{align}
The label $m$ will be referred as occupation number of the 
$q$-oscillator.

Consider the maps 
$\pi_i\,(1 \le i \le n-1):\, A_q(\mathrm{SL}_n) \rightarrow \mathrm{Osc}_q$ 
that send the generators $T = (t_{ij})$ as
\begin{equation}\label{repA}
\begin{small}
\begin{pmatrix}
\!\! \!\!\!\!t_{11} &  &  &  &  &t_{1n}\!\!\!\!\!\!\\
\qquad\ddots  &  & & &&\\
& \!\!\!\!\!\!\!\!\!\!\!\! t_{i-1,i-1}&&&  &   \\
&\qquad \!\!\!\! t_{i,i} & t_{i,i+1}  &&&\\
&\qquad t_{i+1,i}  & t_{i+1,i+1}  & &&\\
&&& \!\!\!\!\!\!\!\! t_{i+2,i+2}&&&\\
&&&& \!\!\!\!\!\!\ddots & \\
\!\! \!\!\!\!t_{n1} &&&&& t_{nn}\!\!\!\!\!\!
\end{pmatrix}
\mapsto
\begin{pmatrix}
\sigma_{i}\bf{1} &  &  &  &  &&&\\
& \ddots & &  & & &&\\
 &  & \sigma_{i}\bf{1} &&&  &  & \\
&&& \mu_i {\rm{\bf a}}^- & \alpha_i {\rm{\bf k}} &&&\\
&&& \beta_i {\rm{\bf k}}& \nu_i {\rm{\bf a}}^+& &&\\
&&&&&\sigma_{i}\bf{1}&&\\
&&&&&& \ddots & \\
&&&&&&& \sigma_{i}\bf{1}
\end{pmatrix}
\end{small},
\end{equation}
where all the blanks on the RHS mean $0$.
$\alpha_i, \beta_i, \mu_i, \nu_i$ and $\sigma_i$ are parameters obeying the relation
\begin{equation}\label{paraA}
\alpha_i\beta_i = -q\sigma_i^{-n+2},\quad
\mu_i\nu_i = \sigma_i^{-n+2}.
\end{equation}
It is elementary to show
\begin{proposition}\label{pr:piA}
The maps $\pi_1, \ldots, \pi_{n-1}$ are naturally extended to the 
algebra homomorphisms. 
The resulting representations of $A_q(\mathrm{SL}_n)$ on $\mathcal{F}_q$ are irreducible.
\end{proposition}

Let
\begin{equation}\label{Pt}
P(x\otimes y) = y\otimes x
\end{equation}
be the transposition, where $x$ and $y$ are taken from
$\mathrm{Osc}_q$ (or its representation $\mathrm{End}(\mathcal{F}_q)$).
For $|i-j| \ge 2$, one can check 
$P (\pi_i \otimes \pi_{j})(\Delta(f))
=  (\pi_{j} \otimes \pi_{i})(\Delta(f))P$ for any $f \in A_q(\mathrm{SL}_n)$.
This is due to the $f=t_{km}$ case
\begin{equation}\label{tp2}
P\Bigl(\sum_l \pi_i(t_{kl})\otimes \pi_j(t_{lm})\Bigr) 
= \Bigl(\sum_l \pi_j(t_{kl})\otimes \pi_i(t_{lm})\Bigr) P \;\quad
\text{for }\; |i-j|\ge 2
\end{equation}
for any $k$ and $m$.
The point here is that the naively obtained expression 
$(\sum_l \pi_j(t_{lm})\otimes \pi_i(t_{kl})) P$ 
violates the coproduct structure (\ref{copro}),
but it equals to the RHS of (\ref{tp2}) thanks to the 
simple structure of the matrix (\ref{repA}).

In what follows, we set $\sigma_i = 1$ for all $i$, which 
does not cause an essential loss of generality.
See Remark \ref{re:sigA}. 
The representations 
$A_q(\mathrm{SL}_n) \rightarrow \mathrm{End}(\mathcal{F}_q)$
defined by (\ref{akf}), (\ref{repA}) with
\begin{equation}\label{paraAA}
\alpha_i\beta_i = -q,\quad
\mu_i\nu_i = 1
\end{equation}
will also be denoted by $\pi_i=\pi_i^{\alpha_i, \mu_i}$.

Let $W(\mathrm{SL}_n) = \langle s_1,\ldots, s_{n-1}\rangle$ be the Weyl group
of $\mathrm{SL}_n$.
It is a Coxeter system 
with generators $s_1,\ldots, s_{n-1}$ obeying the relations
\begin{equation*}
s^2_i = 1, \quad \;\;s_is_j = s_js_i\; \;(|i-j|\ge 2), \quad\;\;
s_is_js_i = s_js_is_j\; \;(|i-j|=1).
\end{equation*}
We will often abbreviate 
$\pi_{i_1}\otimes \pi_{i_2} \otimes \cdots \otimes \pi_{i_r}$
to $\pi_{i_1,i_2,\ldots, i_r}$ in the sequel.
These indices should not be confused with the ones
appearing later   
signifying the {\em positions} in the multiple tensor products.

\begin{theorem}[\cite{So1, So2}]\label{th:so}

(i) The $A_q(\mathrm{SL}_n)$-module 
$\pi_{i_1,i_2,\ldots, i_r}$ is irreducible if 
$s_{i_1}s_{i_2} \cdots s_{i_r}$ is a reduced expression
 of an element from $W(\mathrm{SL}_n)$.

(ii) If $s_{i_1}s_{i_2} \cdots s_{i_r} = s_{j_1}s_{j_2} \cdots s_{j_r}$ are 
two reduced expressions, then the two irreducible representations
$\pi_{i_1,i_2,\ldots, i_r}$ and $\pi_{j_1,j_2,\ldots, j_r}$ are equivalent.
\end{theorem}

The same theorem holds also for $A_q(G)$ 
for any simple Lie group $G$,
where $\pi_i$ is associated to each node $i$ 
of the Dynkin diagram \cite{So1,So2}.

\subsection{\mathversion{bold}$A_q(\mathrm{SL}_3)$ and intertwiner}

The isomorphism of the two irreducible representations will be called 
{\em intertwiner}.
By Schur's lemma, it is unique up to an overall normalization.
The first nontrivial situation arises in $A_q(\mathrm{SL}_3)$,
where one has the equivalence 
$\pi_{121} \simeq \pi_{212}$
reflecting the Coxeter relation $s_1s_2s_1 = s_2s_1s_2$.
Let 
\begin{align}\label{phdefA}
\Phi : 
\mathcal{F}_{q} \otimes \mathcal{F}_{q}\otimes \mathcal{F}_{q}\longrightarrow
\mathcal{F}_{q} \otimes \mathcal{F}_{q}\otimes \mathcal{F}_{q}
\end{align}
be the associated intertwiner.
It is characterized by the relations:
\begin{align}
&\pi_{212}(\Delta(f))\circ \Phi = \Phi \circ \pi_{121}(\Delta(f))
\quad (\forall f \in A_q(\mathrm{SL}_3)),\label{pheqA}\\
&\Phi (|0\rangle \otimes|0\rangle \otimes|0\rangle) 
=|0\rangle \otimes|0\rangle \otimes|0\rangle, \label{phno}
\end{align}
where the latter just fixes a normalization.
As in the case of quantum $R$ matrices,
we find it convenient to work with $R$ defined by
\begin{align}\label{rphA13}
R = \Phi  P_{13} : 
\; \mathcal{F}_{q} \otimes \mathcal{F}_{q}\otimes \mathcal{F}_{q}
\longrightarrow 
\mathcal{F}_{q} \otimes \mathcal{F}_{q}\otimes \mathcal{F}_{q},
\end{align}
where $P_{13}(x\otimes y \otimes z) = z \otimes y \otimes x$.
The intertwining relation for $R$ reads
\begin{align}
&\pi_{212}(\Delta(f))\circ R = R \circ \pi_{121}(\tilde{\Delta}(f))
\quad (\forall f \in A_q(\mathrm{SL}_3)),\label{reRA}
\end{align}
where $\tilde{\Delta}(f) = P_{13}(\Delta(f))P_{13}$, namely,
\begin{equation}\label{gyaku}
\tilde{\Delta}(t_{ij}) = \sum_{l_1,l_2} 
t_{l_2 j} \otimes t_{l_1l_2} \otimes t_{i l_1}.
\end{equation}
For instance, the equation (\ref{reRA}) with $f=t_{11}$ gives
\begin{equation}\label{eq11} 
\mu_1({\bf 1}\otimes {\rm{\bf a}}^- \otimes {\bf 1}) R = 
R \bigl(\mu_1^2({\rm{\bf a}}^-\otimes {\bf 1}\otimes {\rm{\bf a}}^-)
-q\mu_2({\rm{\bf k}}\otimes {\rm{\bf a}}^-\otimes {\rm{\bf k}})\bigr).
\end{equation}
The $R$ is regarded as a matrix 
$R = (R^{abc}_{ijk})$ whose elements are specified by
\begin{align}\label{rmel}
R(|i\rangle \otimes |j\rangle \otimes |k\rangle) = 
\sum_{a,b,c} R^{a b c}_{i j k}
|a\rangle \otimes |b\rangle \otimes |c\rangle.
\end{align}
The normalization condition (\ref{phno}) becomes 
$R^{000}_{000}=1$.
Introduce the notations
\begin{align*}
(q)_i = \prod_{j=1}^i(1-q^j),\quad
\left\{i_1,\ldots, i_r \atop j_1, \ldots, j_s\right\} =
\begin{cases} 
\frac{\prod_{k=1}^r(q^2)_{i_k}}{\prod_{k=1}^s(q^2)_{j_k}} & 
\forall i_k, j_k \in \Z_{\ge 0},\\
0 & \text{otherwise}.
\end{cases}
\end{align*}
Except the dependence on the parameters $\mu_i$,
the following is essentially due to \cite{KV}. 
\begin{theorem}\label{th:agata}
The equation (\ref{reRA}) with 
$R^{000}_{000} = 1$ has a unique solution. It has the form
\begin{align}
R^{abc}_{ijk} &= \mu_1^{a-j+k}\mu_2^{b-a-k}
\mathscr{R}^{abc}_{ijk},\label{rA1}
\end{align}
where $\mathscr{R}^{abc}_{ijk}$ is independent of the parameters and given by
\begin{align}
\mathscr{R}^{abc}_{ijk} &=\delta_{i+j,a+b}\delta_{j+k,b+c}
\sum_{\lambda+\mu=b}(-1)^\lambda
q^{i(c-j)+(k+1)\lambda+\mu(\mu-k)}
\left\{{i,j,c+\mu \atop \mu,\lambda,i-\mu,j-\lambda,c}\right\},
\label{rA2}
\end{align}
where the sum is over $\lambda, \mu \in \Z_{\ge 0}$ such that $\lambda + \mu = b$.
\end{theorem}

\begin{proof}
For the choices $f=t_{13}$ and $t_{31}$, 
(\ref{reRA}) reads
\begin{equation*}
(q^{b+c}-q^{j+k})R^{abc}_{ijk}=0,
\quad
(q^{a+b}-q^{i+j})R^{abc}_{ijk}=0
\end{equation*}
giving the factor 
$\delta_{i+j,a+b}\delta_{j+k,b+c}$ in (\ref{rA2}).
We call this support property {\em conservation law}.
Thanks to it, one readily checks that 
$\mathscr{R}^{abc}_{ijk}$ defined in (\ref{rA1}) 
obeys the recursion relations independent of the parameters.  
Let us pick the following two among them:
\begin{align}
f=t_{32}: & \; \mathscr{R}^{abc}_{ijk} = (1-q^{2c+2})q^{a-j}
\mathscr{R}^{a,b-1,c+1}_{i-1,j,k}+q^{c-j}\mathscr{R}^{a-1,b,c}_{i-1,j,k},\label{t32}\\
f=t_{33}: & \; \mathscr{R}^{abc}_{ijk} = 
\mathscr{R}^{a-1,b,c-1}_{i,j-1,k}-q^{a+c+1}\mathscr{R}^{a,b-1,c}_{i,j-1,k}.\label{t33}
\end{align}
They can be iterated $m$ times to reduce $i$ and $j$ indices as
\begin{align*}
\mathscr{R}^{abc}_{ijk} = \delta_{i+j,a+b}\delta_{j+k,b+c}\sum_{r=0}^m
q^{(m-r)(c-j)+r(a-j-m+r)}
\left\{{m, c+r \atop r, m-r,c}\right\}
\mathscr{R}^{a-m+r,b-r,c+r}_{i-m,j,k},\\
\mathscr{R}^{abc}_{ijk} = \delta_{i+j,a+b}\delta_{j+k,b+c}\sum_{r=0}^m(-1)^r
q^{r(a+c-2m+2r+1)}
\left\{{m \atop r, m-r}\right\}
\mathscr{R}^{a-m+r,b-r,c-m+r}_{i,j-m,k}.
\end{align*}
Combining them, we get (\ref{rA2}).
Since the intertwiner exists, the validity of (\ref{reRA})
for the other $f$'s is guaranteed.
\end{proof}

\begin{proposition}\label{pr:Rinv}
The matrices 
$R$ and $\mathscr{R} = (\mathscr{R}^{abc}_{ijk})$ have the properties
\begin{align}
&R^{-1} = R|_{\mu_1 \leftrightarrow \mu_2},\qquad
\mathscr{R}^{-1} = \mathscr{R},
\label{RAinv}\\
&(q^2)_a(q^2)_b(q^2)_c\,\mathscr{R}^{abc}_{ijk} 
= (q^2)_i(q^2)_j(q^2)_k\,\mathscr{R}^{ijk}_{abc} .\label{Rsym}
\end{align}
\end{proposition}
\begin{proof}
Take $\alpha_i=\beta_i$ in (\ref{repA})
without influencing the relations (\ref{paraAA}), (\ref{rA1}) 
and the parameters $\mu_i$ and $\nu_i$.
Then $\pi_\ell(t_{ij})= \pi_\ell(t_{ji})$ holds for any 
$i, j, \ell$, and therefore 
$\pi_{121}(\tilde{\Delta}(t_{ij})) = \pi_{121}(\Delta(t_{ji}))$ by (\ref{gyaku}).
Thus the defining equations (\ref{reRA})
for $R$ are equivalent to those for $R^{-1}$ with 
the interchange $\mu_1 \leftrightarrow \mu_2$.
This proves the first equality in (\ref{RAinv}), which also 
implies the second one by (\ref{rA1}).
Next to show (\ref{Rsym}) we tune the parameter as
$\mu_1=\mu_2 = \nu_1=\nu_2  =1,\, \alpha_1=\alpha_2$ 
and $\beta_1 = \beta_2$ without violating (\ref{paraAA}).
Then we have $R=\mathscr{R}$ by (\ref{rA1}), hence  
(\ref{Rsym}) is the assertion that 
$\mathcal{D} R$ is symmetric.
Here $\mathcal{D} = D\otimes D\otimes D$ with 
$D \in \mathrm{End}(\mathcal{F}_q)$ being a diagonal operator defined by
\begin{equation}\label{Ddef}
D | m\rangle  = (q^2)_m | m \rangle.
\end{equation}  
From (\ref{akf}) we see 
$({\rm{\bf a}}^{\pm})^T = D  {\rm{\bf a}}^{\mp} D^{-1}$
and 
$({\rm{\bf k}})^T = D  {\rm{\bf k}} D^{-1}$
for the transposed actions on the Fock space.
Under the above choice of the parameters, this leads to
\begin{equation*}
\pi_1(t_{ij})^T = D \pi_2(t_{j' i'}) D^{-1},\quad
\pi_2(t_{ij})^T = D \pi_1(t_{j' i'}) D^{-1}\quad (i' = 4-i)
\end{equation*}
for the single representations (\ref{repA}), therefore 
\begin{equation*}
(\pi_{212}\Delta (t_{ij}))^T = \mathcal{D}
(\pi_{121} \tilde{\Delta}(t_{j' i'})) \mathcal{D}^{-1},\quad
(\pi_{121}\tilde{\Delta} (t_{ij}))^T = \mathcal{D}
(\pi_{212} \Delta(t_{j' i'})) \mathcal{D}^{-1}
\end{equation*}
for their tensor products. See (\ref{gyaku}). 
Taking the transpose of (\ref{reRA}) by means of the above formula, 
one finds that 
$\mathcal{D}^{-1}R^T\mathcal{D}$ again satisfies (\ref{reRA}).
Since the intertwiner is unique up to normalization and $R$ and 
$\mathcal{D}^{-1}R^T\mathcal{D}$ coincide on 
$|0\rangle \otimes |0\rangle \otimes |0\rangle$,
we conclude 
$\mathcal{D}^{-1}R^T\mathcal{D} = R$
hence $(\mathcal{D} R)^T  = \mathcal{D} R$.
\end{proof}

\begin{remark}\label{re:sigA}
If $\sigma_1$ and $\sigma_2$ are retained in (\ref{repA}),
the intertwining relation (\ref{reRA}) has the solution only if 
$\sigma_1=\sigma_2$.
The resulting modification of (\ref{rA1}) is only to multiply 
an overall power of $\sigma_1$ on its RHS.
\end{remark}

The equation (\ref{pheqA}) or (\ref{reRA}) 
have also been considered effectively in several guises and referred 
as tetrahedral Zamolodchikov algebra, local Yang-Baxter equation or 
quantum Korepanov equation, etc.
See for example
\cite{MN89,Korep93,Korep95,KKS98,BS,BMS}.

The result (\ref{rA2}) has been obtained by using (\ref{reRA}) with
$f=t_{32}$ and $t_{33}$.
It is the same route as those taken in \cite{KV} and \cite{KS},
although the formulas therein contain misprints unfortunately.
One can derive apparently different expressions from other choices of $f$.
Here we include a remark on the choice $f=t_{11}$ given in (\ref{eq11}).
In terms of the matrix elements of $\mathscr{R}$ it reads
\begin{equation}
(1-q^{2b})\mathscr{R}^{abc}_{ijk}
=(1-q^{2i})(1-q^{2k})\mathscr{R}^{a,b-1,c}_{i-1,j,k-1}-
q^{i+k+1}(1-q^{2j})\mathscr{R}^{a,b-1,c}_{i,j-1,k}.\label{pr}
\end{equation}
On the other hand, recall the 3D $R$ matrix \cite[eq.(30)]{BS}
whose elements are given by
\begin{equation}\label{231}
\langle i,j,k | {\bf r} | a,b,c \rangle = 
\delta_{i+j,a+b}\delta_{j+k,b+c}
\frac{q^{(a-j)(c-j)}}{(q^2)_{b}}P_{b}(q^{2i},q^{2j},q^{2k}),
\end{equation}
where $P_m$ is determined by the recursion 
\begin{equation}
P_m(x,y,z) = (1-x)(1-z)P_{m-1}(q^{-2}x,y,q^{-2}z)
-q^{2-2m}xz(1-y)P_{m-1}(x,q^{-2}y,z)\label{br}
\end{equation}
and the initial condition $P_0(x,y,z)=1$.
In (\ref{231}) we have removed the power  $q^{-\beta}$ and the 
sign $(-1)^\beta$ in \cite[eq.(30)]{BS} in view of the fact that
the first one is absent in \cite[eq.(59)]{BMS}
and the latter can be absorbed into $\varepsilon$ in \cite[eq.(22)]{BMS}.
We point out that
\begin{align}
&\mathscr{R}^{abc}_{ijk} =\langle i,j,k | {\bf r} | a,b,c \rangle.
\label{pp}
\end{align}
To show this, substitute the RHS of (\ref{231}) 
into (\ref{pr}). Then it agrees with 
(\ref{br}) due to the conservation law $a+b=i+j$ and $b+c=j+k$.
It remains to check the initial condition
\begin{align}
\mathscr{R}^{a 0 c}_{i j k} = \delta_{i+j, a}\delta_{j+k,c}\, q^{(a-j)(c-j)}= 
\delta_{i+j, a}\delta_{j+k,c} \,q^{ik},
\end{align}
which is straightforward by (\ref{rA2}).
Thus (\ref{pp}) has been proved.
Note that the formula (\ref{231}) with (\ref{rA1}) tells 
\begin{equation}\label{swapR}
\mathscr{R}^{abc}_{ijk} = \mathscr{R}^{cba}_{kji},\qquad 
R^{abc}_{ijk} = R^{cba}_{kji}.
\end{equation}

\begin{remark}\label{re:Rq0}
One can show that  $\mathscr{R}^{abc}_{ijk}$ is a
{\em polynomial} in $q$ with integer coefficients.
More precisely, 
$\mathscr{R}^{abc}_{ijk} \in q^{\xi}\Z[q^2] $ where $\xi=0,1$ 
is specified by 
$\xi \equiv (a-j)(c-j)\mod 2$.
The matrix
\begin{equation*}
\mathcal{R} = (\mathcal{R}^{abc}_{ijk}) := \mathscr{R}|_{q=0}
\end{equation*}
has the elements
\begin{equation}\label{rzero}
\mathcal{R}^{abc}_{ijk} = \mathscr{R}^{abc}_{ijk}|_{q=0} = \delta_{i+j,a+b}\delta_{j+k,b+c}
\delta_{i, b+(a-c)_+}\delta_{j, \min(a,c)}\delta_{k, b+(c-a)_+},
\end{equation}
where $(y)_+ = \max(y,0)$.
The proof is far simpler than the analogous result on 
$K$ in Theorem \ref{th:combk} which will be detailed in Appendix \ref{app:kq0}.
Moreover, (\ref{RAinv}) implies $\mathcal{R} = \mathcal{R}^{-1}$.
Thus $\mathcal{R}$ defines a bijection on each finite set 
$\{(a,b,c) \in (\Z_{\ge 0})^3\,| a+c=\text{const}, b+c=\text{const}\}$
characterized by the values of conserved quantities.
We call $\mathcal{R}$ {\em combinatorial} 3D $R$.
It is analogous to the $q=0$ case of quantum $R$ matrices, 
which has led to many applications.
See for example \cite{KO, IKT}.
More remarks are in order in Section \ref{subsec:tropA}.
\end{remark}

\begin{example}\label{ex:rq0}
The following is the list of all the nonzero $\mathscr{R}^{abc}_{314}$.
\begin{align*}
\mathscr{R}^{041}_{314} &= -q^2 (1 - q^4) (1 - q^6) (1 - q^8), \\ 
\mathscr{R}^{132}_{314} &=(1 - q^6) (1 - q^8)(1 - q^4 - q^6 - q^8 - q^{10}),\\
\mathscr{R}^{223}_{314} &= q^2 (1 + q^2) (1 + q^4) (1 - q^6) (1 - q^6 - q^{10}), \\
\mathscr{R}^{314}_{314} &=q^6(1 + q^2 + q^4 - q^8 - q^{10} - q^{12} - q^{14}),\\
\mathscr{R}^{405}_{314} &=q^{12}.
\end{align*}
Thus $\mathcal{R}^{abc}_{314}=\delta_{a,1}\delta_{b,3}\delta_{c,2}$ 
in agreement with (\ref{rzero}).
\end{example}

\subsection{\mathversion{bold}$A_q(\mathrm{SL}_4)$ and tetrahedron equation}\label{ss:tea}

Consider $A_q(\mathrm{SL}_4)$ and let
$\pi_i=\pi_i^{\alpha_i, \mu_i}\, (i=1,2,3)$ be its irreducible representations  
specified around (\ref{paraAA}) following \cite{KV}.
The tensor products $\pi_{212}$ and $\pi_{121}$ 
are intertwined by the same 
$\Phi$ as the one for $A_q(\mathrm{SL}_3)$ given in (\ref{rphA13}) and 
Theorem \ref{th:agata}.
Write it as $\Phi^{(1)}$, which involves the parameters $\mu_1$ and $\mu_2$.
It is easy to see that 
$\pi_{323}$ and $\pi_{232}$ is similarly intertwined  
by $\Phi^{(2)}$ obtained from  $\Phi^{(1)}$ 
by changing the parameters $(\mu_1,\mu_2)$ to $(\mu_2,\mu_3)$. 
Thus we have
\begin{equation}\label{pheqAA}
\begin{split}
&\pi_{212}(\Delta(f))\circ \Phi^{(1)}
= \Phi^{(1)}\circ \pi_{121}(\Delta(f)),\\
&\pi_{323}(\Delta(f))\circ \Phi^{(2)}
= \Phi^{(2)}\circ \pi_{232}(\Delta(f))
\end{split}
\end{equation}
for any $f \in A_q(\mathrm{SL}_4)$.
According to (\ref{rphA13}), we set
\begin{equation}\label{R12p}
\Phi^{(1)} = R^{(1)} P_{13},\quad 
\Phi^{(2)} = R^{(2)} P_{13},
\end{equation}
where $R^{(1)}$ is (\ref{rA1}) and 
$R^{(2)}$ is obtained from it by changing 
 $(\mu_1,\mu_2)$ to $(\mu_2,\mu_3)$.
The cumbersome upper indices can always be forgotten
by specializing the parameters to $\mu_1=\mu_2=\mu_3$.

Let $w_0 \in W(\mathrm{SL}_4)$ be the longest element of the Weyl group.
We pick two reduced expressions say,
\begin{equation}\label{red6}
w_0 = s_1s_2s_3s_1s_2s_1
= s_3s_2s_3s_1s_2s_3,
\end{equation} 
where the two sides are interchanged by 
replacing $s_i$ by $s_{4-i}$ and reversing the order.
According to Theorem \ref{th:so}, we have the equivalence of the 
two representations of $A_q(\mathrm{SL}_4)$:
\begin{equation}\label{pi6}
\pi_{123121} \simeq \pi_{323123}.
\end{equation}
Let $P_{ij}$ and $\Phi^{(1)}_{ijk}, \Phi^{(2)}_{ijk}$ 
be the transposition $P$ (\ref{Pt}) 
and the intertwiners $\Phi^{(1)}, \Phi^{(2)}$
that act on the tensor components specified by the indices.
These components must be adjacent  
(i.e. $j-i=k-j=1$) to make the
relations (\ref{tp2}) and (\ref{pheqAA}) work.
With this guideline, one can construct the 
intertwiners for (\ref{pi6}) by following the 
transformation of the reduced expressions 
by the Coxeter relations 
\begin{equation*}
s_1s_3=s_3s_1,\quad s_1s_2s_1=s_2s_1s_2,\quad s_2s_3s_2 = s_3s_2s_3.
\end{equation*}
There are two ways to achieve this.
In terms of the indices, they look as follows:
\begin{alignat}{3}
123\underline{121}\quad & \Phi^{(1)}_{456} \qquad\qquad 
&12\underline{31}21\quad & P_{34} \nonumber\\
1\underline{232}12 \quad & \Phi^{(2)}_{234} \qquad\qquad 
&\underline{121}321\quad & \Phi^{(1)}_{123} \nonumber\\
\underline{13}2\underline{31}2 \quad & P_{12}P_{45} \qquad\qquad 
&21\underline{232}1\quad & \Phi^{(2)}_{345} \nonumber\\
3\underline{121}32 \quad & \Phi^{(1)}_{234} \qquad\qquad 
&2\underline{13}2\underline{31}\quad & P_{23}P_{56} \nonumber\\
321\underline{232} \quad & \Phi^{(2)}_{456} 
\qquad\qquad &23\underline{121}3\quad & \Phi^{(1)}_{345} \nonumber\\
32\underline{13}23 \quad & P_{34} \qquad\qquad 
&\underline{232}123 \quad & \Phi^{(2)}_{123}\nonumber\\
323123 \quad &  \qquad\qquad &323123\quad & \label{2ways}
\end{alignat}
The underlines indicate the components to which the intertwiners given on the right 
are to be applied.
Since (\ref{pi6}) is irreducible, we get 
\begin{equation}\label{phpA}
P_{34}\Phi^{(2)}_{456} \Phi^{(1)}_{234}P_{12}P_{45}
\Phi^{(2)}_{234} \Phi^{(1)}_{456}\\
= \Phi^{(2)}_{123}\Phi^{(1)}_{345} P_{23}P_{56}\Phi^{(2)}_{345}
\Phi^{(1)}_{123}P_{34}.
\end{equation}
Substituting 
$\Phi^{(1)}_{ijk} = R^{(1)}_{ijk} P_{jk}$ and 
$\Phi^{(2)}_{ijk} = R^{(2)}_{ijk} P_{jk}$\footnote{See (\ref{R12p}). 
Indices $ijk$ of $R$ here signify the tensor components
and should not be confused with those for the matrix elements (\ref{rmel}).}
into this and sending all the $P_{ij}$'s through to the right, 
we find that the products of $P_{ij}$'s
correspond to the longest element in the symmetric group $\mathfrak{S}_6$ 
on the both sides. Thus canceling them out, we obtain
\begin{equation}\label{te1}
R^{(2)}_{356}R^{(1)}_{246}R^{(2)}_{145}R^{(1)}_{123}
=R^{(2)}_{123}R^{(1)}_{145}R^{(2)}_{246}R^{(1)}_{356}
\end{equation}
for the operators acting on $(\mathcal{F}_q)^{\otimes 6}$.
When $\mu_1=\mu_2=\mu_3$ hence the upper indices can be  removed,
it reproduces a version of the Zamolodchikov tetrahedron equation 
\cite{Zam80,Zam81}.
In particular it implies that $\mathscr{R}$ (\ref{rA2}) 
satisfies 
\begin{equation}\label{tehat}
\mathscr{R}_{356}\mathscr{R}_{246}\mathscr{R}_{145}\mathscr{R}_{123}
=\mathscr{R}_{123}\mathscr{R}_{145}\mathscr{R}_{246}\mathscr{R}_{356}.
\end{equation}
This is an equality among polynomials of $q$ and free from the other parameters.

Let us write $123121 \rightarrow 323123$ to stand for the above
calculation leading to the tetrahedron equation (\ref{te1}).
There are 16 reduced expressions for $w_0$ in total and 
one can play the same game with the other 7 pairs.
The result is given by
\begin{align*}
121321\rightarrow 321323 : \quad 
&R^{(2)}_{456}R^{(1)}_{236}R^{(2)}_{135}R^{(1)}_{124}
=R^{(2)}_{124}R^{(1)}_{135}R^{(2)}_{236}R^{(1)}_{456},\\
123212 \rightarrow 232123: \quad
&\bar{R}^{(2)}_{321}R^{(2)}_{156}R^{(1)}_{246}R^{(2)}_{345}
=  R^{(1)}_{345}R^{(2)}_{246}R^{(1)}_{156}\bar{R}^{(1)}_{321},\\
132132 \rightarrow 213213 : \quad
&R^{(2)}_{346}R^{(1)}_{126}\bar{R}^{(1)}_{531}\bar{R}^{(2)}_{542}
= \bar{R}^{(1)}_{542}\bar{R}^{(2)}_{531}
R^{(2)}_{126}R^{(1)}_{346},\\
132312 \rightarrow 231213 : \quad
&R^{(2)}_{246}R^{(1)}_{136}\bar{R}^{(1)}_{521}\bar{R}^{(2)}_{543}
= \bar{R}^{(1)}_{543}\bar{R}^{(2)}_{521}
R^{(2)}_{136}R^{(1)}_{246},\\
212321 \rightarrow 321232 : \quad
&R^{(1)}_{234}R^{(2)}_{135}R^{(1)}_{126}\bar{R}^{(1)}_{654}
=\bar{R}^{(2)}_{654}R^{(2)}_{126}R^{(1)}_{135}R^{(2)}_{234},\\
213231 \rightarrow 312132 : \quad
&R^{(2)}_{135}R^{(1)}_{146}\bar{R}^{(1)}_{652}\bar{R}^{(2)}_{432}
=\bar{R}^{(1)}_{432}\bar{R}^{(2)}_{652}R^{(2)}_{146}R^{(1)}_{135},\\
231231 \rightarrow 312312 : \quad
&R^{(2)}_{134}R^{(1)}_{156}\bar{R}^{(1)}_{642}\bar{R}^{(2)}_{532}
=\bar{R}^{(1)}_{532}\bar{R}^{(2)}_{642}R^{(2)}_{156}R^{(1)}_{134},
\end{align*}
where the notation $\bar{R} = R^{-1}$
has been used to uniform the spacing.
Thanks to (\ref{swapR}), we have 
$R^{(1)}_{ijk} = R^{(1)}_{kji}$ and 
$R^{(2)}_{ijk} = R^{(2)}_{kji}$.
Using this symmetry it can be checked that
all the above relations reduce to the single tetrahedron equation
(\ref{te1}).

One can derive similar compatibility conditions for 
the intertwiners in $A_q(\mathrm{SL}_n)$ with $n \ge 5$. 
We expect that they are all attributed to  (\ref{tehat})
in the parameter-free case.
For instance for $n=5$, the longest element  is of length 10 and 
the compatibility is expressed as
\begin{equation*}
\mathscr{R}_{123}\mathscr{R}_{145}\mathscr{R}_{246}\mathscr{R}_{356}
\mathscr{R}_{178}\mathscr{R}_{279}\mathscr{R}_{389}\mathscr{R}_{470}
\mathscr{R}_{580}\mathscr{R}_{690} = 
\text{product in reverse order}, 
\end{equation*}
where $0$ is the abbreviation of $10$.
This can be derived by using (\ref{tehat}) five times.

Setting $q=0$ in (\ref{tehat}), we find that the combinatorial 3D $R$ 
in Remark \ref{re:Rq0}  also satisfies the tetrahedron equation
\begin{equation}\label{tecr}
\mathcal{R}_{356}\mathcal{R}_{246}\mathcal{R}_{145}\mathcal{R}_{123}
=\mathcal{R}_{123}\mathcal{R}_{145}\mathcal{R}_{246}\mathcal{R}_{356}.
\end{equation}
It is an identity of the bijections on subsets of $(\Z_{\ge 0})^6$.

\begin{example}\label{ex:cr}
To demonstrate (\ref{tecr}), 
we denote a monomial $|i_1\rangle \otimes \cdots \otimes | i_6\rangle 
\in (\mathcal{F}_q)^{\otimes 6}$ simply by $|i_1,\ldots, i_6\rangle$.
Then the monomial, say, $|314516\rangle$ is transformed as in 
Figure \ref{fig:cte}. (In this example $i_1,\ldots, i_6$ remain less than ten,
so they are all specified by a single digit.) 

\begin{figure}[h]
\begin{picture}(100,135)(10,0)
\put(50,120){$|314516\rangle$}

\put(7,110){$\mathcal{R}_{123}\; \swarrow$}
\put(95,110){$\searrow \,\mathcal{R}_{356}$}

\put(0,90){$|132516\rangle$} \put(100,90){$|311543\rangle$} 

\put(-7,75){$\mathcal{R}_{145} \downarrow$}
\put(115,75){$\downarrow \mathcal{R}_{246}$}

\put(0,60){$|532156\rangle$} \put(100,60){$|351147\rangle$} 

\put(-7,45){$\mathcal{R}_{246} \downarrow$}
\put(115,45){$\downarrow \mathcal{R}_{145}$}

\put(0,30){$|512354\rangle$} \put(100,30){$|151327\rangle$} 

\put(7,10){$\mathcal{R}_{356}\searrow$}
\put(95,10){$\swarrow \,\mathcal{R}_{123}$}

\put(50,0){$|515327\rangle$}
\end{picture}    
\caption{An example of tetrahedron equation 
(\ref{tecr}) for combinatorial 3D $R$.}
\label{fig:cte}
\end{figure}

\noindent
The first SW arrow  by $\mathcal{R}_{123}$ is 
due to Example \ref{ex:rq0}.
If one keeps $q$ generic and lets (\ref{tehat}) 
act on the same monomial $|314516\rangle$,
each side generates 300 monomials. 
\end{example}

\subsection{\mathversion{bold}Classical aspects and triad of 3D $R$}
\label{subsec:tropA}
In terms of $\Phi$ (\ref{rphA13}), 
the combinatorial 3D $R$  is rephrased as the following 
map on $(\Z_{\ge 0})^3$:
\begin{equation*}
\Phi|_{q=0}: (a,b,c) \mapsto (b+c-\min(a,c), \min(a,c), a+b-\min(b,c)).
\end{equation*}
The same map has been introduced in 
\cite[p451]{L}  and independence of its 
compositions corresponding to any reduced expressions of the 
longest element of $W(\mathrm{SL}_n)$ was utilized.

By regarding $a,b,c$ as indeterminates, 
this property generalizes to the birational map 
$(a,b,c) \mapsto \bigl(\frac{bc}{a+c}, \,a+c,\, \frac{ab}{a+c}\bigr)$.
The previous one is reproduced  
via the {\em ultradiscretization} (or tropical variable change)
$a b \rightarrow a+b$ and $a+b \rightarrow \min(a,b)$ 
as pointed out by \cite{YY}. See also \cite{BFZ}.
Its composition with $P_{13}$ is the map
\begin{equation}\label{birR}
{\bf R}: (c,b,a) \mapsto 
(\tilde{a},\tilde{b},\tilde{c})= \left(\frac{bc}{a+c}, \,a+c,\, \frac{ab}{a+c}\right)
\end{equation}
to be called the {\em birational} 3D $R$ in the context of 
the present paper.
It is characterized as the unique solution to the matrix equation
\begin{equation}\label{GGG}
G_i(a)G_j(b)G_i(c) = G_j(\tilde{a})G_i(\tilde{b})G_j(\tilde{c}) \qquad (|i-j|=1),
\end{equation}
where $G_i(x)= 1+ x E_{i,i+1}$ is a generator of the 
unipotent subgroup of $\mathrm{SL}_n$.
The ${\bf R}$ is birational due to ${\bf R}^{-1} = {\bf R}$.
The intertwining relation (\ref{pheqA}) 
is a quantization of (\ref{GGG}) (with $(i,j)=(1,2)$).
Note that $G_i(a)G_j(b) = G_j(b)G_i(a)$ for $|i-j|>1$ also holds 
analogously to the Coxeter relations.

Given a Weyl group element $w  \in W(\mathrm{SL}_n)$ (not necessarily longest),
assign the matrix $M=G_{i_1}(x_1)\cdots G_{i_r}(x_r)$ to a
reduced expression $w = s_{i_1}\cdots s_{i_r} $.
Then to any reduced expression 
$w = s_{j_1}\cdots s_{j_r} $ one can assign the
expression $M=G_{j_1}(\tilde{x}_1)\cdots G_{j_r}(\tilde{x}_r)$,
where $\tilde{x}_k$ is determined independently of the 
intermediate steps applying (\ref{GGG}).
This property is the source of the tetrahedron equation for 
${\bf R}$ and forms a classical (or birational) counterpart of 
the previous calculation  (\ref{2ways}).
In fact,  the uniqueness of the map 
$(a,b,c,d,e,f) \mapsto (\tilde{a}, \tilde{b}, \tilde{c}, \tilde{d}, \tilde{e}, \tilde{f})$ 
defined by 
\begin{equation}\label{Gpro}
G_1(a)G_2(b)G_3(c)G_1(d)G_2(e)G_1(f) = 
G_3(\tilde{a})G_2(\tilde{b})G_3(\tilde{c})
G_1(\tilde{d})G_2(\tilde{e})G_3(\tilde{f})
\end{equation}
proves the {\em birational tetrahedron equation}
\begin{equation}\label{bth}
{\bf R}_{356}{\bf R}_{246}{\bf R}_{145}{\bf R}_{123}
={\bf R}_{123}{\bf R}_{145}{\bf R}_{246}{\bf R}_{356},
\end{equation}
where ${\bf R}_{ijk}$ is the one acting on the $i,j,k$-th components
in an array of 6 variables.
This is a version of the so called 
functional tetrahedron equation \cite{Ka96, S98, KKS98}, 
which is known to allow a more general solution than (\ref{birR}) 
connected to the star-triangle electric circuits transformation.

We have summarized the triad of the 3D $R$'s  in Table 1 in Section \ref{sec:intro}.
The tetrahedron equations satisfied by them are given in
(\ref{tehat}), (\ref{tecr}) and (\ref{bth}).
The combinatorial one ${\mathcal R}$ shows up either at 
$q = 0$ of the quantum one or  
ultradiscretization of the birational one.
This is a quite analogous feature to 2D. 
See \cite{KO,IKT} for example. 
In the next section, we will add  
a parallel story corresponding to the third row of the table 1. 

\newpage
\section{$\mathrm{Sp}$ case}\label{sec:C}

\subsection{\mathversion{bold}Quantized algebra of functions 
$A_q(\mathrm{Sp}_{2n})$}
We define $A_q(\mathrm{Sp}_{2n})$ following \cite{RTF}, where it was denoted by 
$\mathrm{Fun(Sp}_q(n))$.
First, we introduce the structure constants 
$(R_{ij, kl})_{1\le i,j,k,l \le 2n}$ and 
$C=-C^{-1}=(C_{ij})_{1 \le i,j \le 2n}$ by
\begin{align*}
\begin{split}
&\sum_{i,j,k,l}R_{ij,kl}E_{ik}\otimes E_{jl}=
q\sum_i E_{ii}\otimes E_{ii} + \sum_{i\neq j, j'}E_{ii}\otimes E_{jj}
+q^{-1}\sum_i E_{ii}\otimes E_{i'i'}\\
&\qquad\qquad\qquad+(q-q^{-1})\sum_{i>j}E_{ij}\otimes E_{ji}
-(q-q^{-1})\sum_{i>j}\epsilon_i\epsilon_jq^{\varrho_i-\varrho_j}E_{ij}\otimes E_{i'j'},
\end{split}\\
&C_{ij}= \delta_{i,j'}\epsilon_i q^{\varrho_j},\quad i' = 2n+1-i,
\quad
\epsilon_i = 1\,(1\le i \le n),\;\;\epsilon_i = -1\,(n<i \le 2n),\\
&(\varrho_1,\ldots, \varrho_{2n}) = (n, n-1, \ldots,1,,-1,\ldots, -n+1, -n).
\end{align*}
Here the indices are summed over 
$\{1,2,\ldots, 2n\}$ under the specified conditions.
The constant $R_{ij,kl}$ is extracted as 
\begin{align*}
&\sum_{1\le i,j,m,l\le 2n}R_{ij,ml}E_{im}\otimes E_{jl}
= q\lim_{x\rightarrow \infty}x^{-2}R(x)|_{k=q^{-1}},
\end{align*}
from the quantum $R$ matrix $R(x)$ \cite{Baz,J2} for 
the vector representation of $U_q(C^{(1)}_n)$ given in \cite[eq.(3.6)]{J2}.
For example the matrix $C$ for $n=2$ reads
\begin{equation*}
C = \left(
\begin{array}{cccc}
 0 & 0 & 0 & q^{-2} \\
 0 & 0 & q^{-1} & 0 \\
 0 & -q & 0 & 0 \\
 -q^2 & 0 & 0 & 0
\end{array}
\right).
\end{equation*}

The quantized algebra of functions $A_q(\mathrm{Sp}_{2n})$
is a Hopf algebra \cite{Abe} generated by $T = (t_{ij})_{1\le i,j\le 2n}$ 
with the relations (\ref{re1}) and 
\begin{align}
&TCT^tC^{-1} = CT^tC^{-1}T = I,\;{\rm i.e. }\;
\sum_{jkl}C_{jk}C_{lm}t_{ij}t_{lk} = \sum_{jkl}C_{ij}C_{kl}t_{kj}t_{lm} = -\delta_{im}.
\label{re2}
\end{align}
The coproduct is again given by (\ref{copro}).
We omit the antipode and counit for they will not be used in this paper.

\subsection{\mathversion{bold}Representations of $A_q(\mathrm{Sp}_6)$}

In $\mathrm{SL}$ case, we first considered 
$A_q(\mathrm{SL}_3)$ to determine the 
intertwiner and then proceeded to $A_q(\mathrm{SL}_4)$ to derive the tetrahedron equation
for the purpose of exposition.
Here we shorten our presentation 
by skipping  $A_q(\mathrm{Sp}_4)$ and considering $A_q(\mathrm{Sp}_6)$ from the outset 
since the latter includes the former and presents a generic situation.

Let $\mathrm{Osc}_q$ and $\mathcal{F}_q$ be the $q$-oscillator algebra 
and the Fock space introduced in (\ref{qoa}) and (\ref{akf}).
For distinction we write 
$\mathrm{Osc}_{q^2} 
= \langle {{\bf 1}, \rm{\bf A}}^+, {\rm{\bf A}}^-, {\rm{\bf K}} \rangle$, 
which acts on $\mathcal{F}_{q^2}$.
Set 
\begin{equation}\label{qi}
q_1 = q,\quad q_2= q,\quad q_3=q^2.
\end{equation}
Consider the maps 
$\pi_i\,(i=1,2,3):\; A_q(\mathrm{Sp}_6) \rightarrow \mathrm{Osc}_{q_i}$ which send the generators 
\begin{equation*}
\begin{pmatrix}
t_{11} & t_{12} & t_{13} & t_{14} & t_{15} & t_{16}\\
t_{21} & t_{22} & t_{23} & t_{24} & t_{25} & t_{26}\\
t_{31} & t_{32} & t_{33} & t_{34} & t_{35} & t_{36}\\
t_{41} & t_{42} & t_{43} & t_{44} & t_{45} & t_{46}\\
t_{51} & t_{52} & t_{53} & t_{54} & t_{55} & t_{56}\\
t_{61} & t_{62} & t_{63} & t_{64} & t_{65} & t_{66}
\end{pmatrix}
\end{equation*}
to the following:
\begin{align}
&\pi_1: \; 
\begin{pmatrix}
\mu_1{\rm{\bf a}}^- & \alpha_1{\rm {\bf k}} & 0 & 0& 0& 0\\
\beta_1{\rm {\bf k}}& \nu_1{\rm{\bf a}}^+ & 0 & 0 & 0& 0\\
0 & 0 & \sigma_1{\bf 1} & 0& 0& 0\\
0 & 0 & 0 & \sigma_1^{-1}{\bf 1} & 0 & 0\\
0 & 0 & 0 & 0 & \nu_1^{-1}{\rm{\bf a}}^- & q\beta_1^{-1}{\rm {\bf k}} \\
0 & 0 & 0 & 0 & q\alpha_1^{-1}{\rm {\bf k}} & \mu_1^{-1}{\rm{\bf a}}^+
\end{pmatrix}\quad(\alpha_1\beta_1 = -q \mu_1\nu_1),\label{pi31}\\
&\pi_2: \; 
\begin{pmatrix}
\sigma_2{\bf 1} & 0 & 0 & 0 & 0 & 0\\
0 & \mu_2{\rm{\bf a}}^- & \alpha_2{\rm {\bf k}} & 0 & 0& 0\\
0 & \beta_2{\rm {\bf k}}& \nu_2{\rm{\bf a}}^+ & 0 & 0 & 0\\
0 & 0 & 0 & \nu_2^{-1}{\rm{\bf a}}^- & q\beta_2^{-1}{\rm {\bf k}} & 0\\
0 & 0 & 0 & q\alpha_2^{-1}{\rm {\bf k}} & \mu_2^{-1}{\rm{\bf a}}^+ & 0\\
0 & 0 & 0 & 0 & 0 & \sigma_2^{-1}{\bf 1}
\end{pmatrix}\quad(\alpha_2 \beta_2 = -q \mu_2\nu_2), \label{pi32}\\
&\pi_3: \; 
\begin{pmatrix}
\rho'{\bf 1} & 0 & 0 & 0 & 0 & 0\\
0 & \rho{\bf 1}  & 0 & 0 & 0 & 0\\
0 & 0 & \mu_3{\rm{\bf A}}^- & \alpha_3{\rm{\bf K}} & 0 & 0\\
0 & 0 & \beta_3{\rm{\bf K}} & \mu_3^{-1}{\rm{\bf A}}^+ & 0 & 0\\
0 & 0 & 0 & 0  & \rho^{-1}{\bf 1} & 0\\
0 & 0 & 0 & 0 & 0 & \rho^{\prime -1}{\bf 1}
\end{pmatrix}\quad(\alpha_3 \beta_3 = -q^2),\label{pi33}
\end{align}
where $\alpha_i, \beta_i, \mu_i\,(i=1,2,3),\, \sigma_i, \nu_i\, (i=1,2),\,
\rho$ and $\rho'$ are parameters obeying the constraints 
in the parentheses.
One can directly verify
\begin{proposition}
The maps $\pi_i\, (i=1,2,3)$ are naturally extended to the 
algebra homomorphisms. 
The resulting representations of $A_q(\mathrm{Sp}_6)$ on $\mathcal{F}_{q_i}$ are irreducible.
\end{proposition}

The representations 
$A_q(\mathrm{Sp}_6) \rightarrow \mathrm{End}(\mathcal{F}_{q_i})$
will also be denoted by $\pi_i\, (i=1,2,3)$.

\subsection{Constraints on parameters}

The Weyl group $W(\mathrm{Sp}_6)=\langle s_1, s_2, s_3\rangle$ 
is a Coxeter system generated by simple reflections $s_1, s_2$ and $s_3$
obeying the relations 
\begin{align}\label{coxre}
s_1^2 = s_2^2= s_3^2=1,\quad s_1s_3=s_3s_1,\quad 
s_1s_2s_1 =s_2s_1s_2,\quad 
s_2s_3s_2s_3 = s_3s_2s_3s_2.
\end{align}
Thus, according to Theorem \ref{th:so} (for $\mathrm{Sp}$), 
one expects the equivalence of the
representations 
\begin{align}
\pi_1 \otimes \pi_3 &\simeq \pi_3 \otimes \pi_1,\label{equiv1}\\
\pi_1 \otimes \pi_2 \otimes \pi_1
&\simeq \pi_2 \otimes \pi_1 \otimes \pi_2,\label{equiv2}\\
\pi_2 \otimes \pi_3 \otimes \pi_2 \otimes\pi_3 
&\simeq \pi_3 \otimes \pi_2 \otimes \pi_3 \otimes\pi_2\label{equiv3}
\end{align}
under an appropriate condition on the parameters.

\begin{proposition}\label{pr:para}
(i)  Eq.  (\ref{equiv1}) holds only if
\begin{equation}\label{para13}
\sigma_1 = \pm 1,\quad \rho=\rho'.
\end{equation}

(ii)  Eq. (\ref{equiv2}) holds only if
\begin{equation}\label{para121}
\sigma_1= \sigma_2,\qquad \alpha_1\beta_1 = \alpha_2 \beta_2.
\end{equation}

(iii)  Under the condition (\ref{para13}), 
the equivalence (\ref{equiv3}) holds only if
\begin{equation}\label{para2323}
\rho=\pm 1,\quad \alpha_2\beta_2 = \pm q,
\end{equation}
where the three signs $\pm 1$ in (\ref{para13}) and (\ref{para2323})
can be chosen independently.
\end{proposition}

\begin{proof}
(i) The intertwiner for (\ref{equiv1}) is just the transposition $P$ in (\ref{Pt}).
The relation (\ref{tp2}) with $(i,j)=(1,3)$ 
directly leads to (\ref{para13}).
(ii) and 
(iii)  are derived by investigating the intertwining relations 
(\ref{pheqA}) for $f \in A_q(\mathrm{Sp}_6)$ and (\ref{pip}).
\end{proof}

In view of Proposition \ref{pr:para} and (\ref{pi31})--(\ref{pi33}), 
we set
\begin{equation}\label{para0}
\begin{split}
\pi_1: &\; \sigma_1 = \sigma,\;\; 
\alpha_1\beta_1 = -\varepsilon q,\;\; \mu_1\nu_1 = \varepsilon,\\
\pi_2: &\; \sigma_2 = \sigma,\;\; 
\alpha_2\beta_2 = -\varepsilon q,\;\; \mu_2\nu_2 = \varepsilon,\\
\pi_3: &\; \rho=\rho',\;\; \alpha_3\beta_3 = -q^2
\end{split}
\end{equation}
in the rest of the paper, where the three sign factors 
\begin{equation}\label{sign}
\varepsilon = \pm 1,\quad \sigma = \pm 1,\quad \rho = \pm 1
\end{equation} 
can be chosen independently. 
Given $(\varepsilon, \sigma, \rho) \in \{\pm1\}^3$,
each $\pi_i$ should be understood as 
the representation containing two independent parameters $\alpha_i$ and $\mu_i$:
\begin{equation}\label{piam}
\pi_i = \pi_i^{\alpha_i, \mu_i}\,:\; A_q(\mathrm{Sp}_6) \rightarrow \mathrm{End}(\mathcal{F}_{q_i})\quad
(i=1,2,3),
\end{equation}
which is defined by (\ref{pi31})--(\ref{pi33}) with (\ref{para0}).

\subsection{\mathversion{bold}Intertwiner $\Phi$ and $R$}\label{sec:R}

Let $\Phi$ be the intertwiner for (\ref{equiv2}).
It is characterized by formally the same relations as (\ref{pheqA})
and (\ref{phno}) with $f \in A_q(\mathrm{Sp}_6)$.
As in the $\mathrm{SL}$ case (\ref{rphA13}), we introduce $R=\Phi P_{13}$  
which satisfies (\ref{reRA}) for $f \in A_q(\mathrm{Sp}_6)$.
It is easy to show

\begin{theorem}
The $R=(R^{abc}_{ijk})$ is given by
\begin{align}
R^{abc}_{ijk} &=\varepsilon^j
(\sigma \mu_1)^{a-j+k}(\sigma \mu_2)^{b-a-k}
\mathscr{R}^{abc}_{ijk},\label{r1}
\end{align}
where $\mathscr{R}^{abc}_{ijk}$ is the parameter-free (except $q$) 
one specified in (\ref{rA2}).
\end{theorem}

Thus the intertwiner $R$ is the same as the $\mathrm{SL}$ case
up to an overall factor.
It satisfies the tetrahedron equation (\ref{te1}) if one identifies
$R^{(1)}$ with (\ref{r1}) and set
$R^{(2)}=R^{(1)}|_{(\mu_1,\mu_2)\rightarrow (\mu_2,\kappa)}$
for any parameter $\kappa$.

Set $\bar{R}=R^{-1}=(\bar{R}^{abc}_{ijk})$. 
From (\ref{RAinv}) and (\ref{r1}), its matrix elements are given by 
\begin{equation}\label{rinv}
\bar{R}^{abc}_{ijk}
=\varepsilon^b(\sigma \mu_1)^{b-a-k}(\sigma \mu_2)^{a-j+k}
\mathscr{R}^{abc}_{ijk}.
\end{equation}

\subsection{\mathversion{bold}Intertwiner $\Psi$ and $K$}\label{sec:K}

Now we face the new object. 
Let
\begin{align}\label{psdef}
\Psi : 
\mathcal{F}_{q} \otimes \mathcal{F}_{q^2}\otimes \mathcal{F}_{q}\otimes \mathcal{F}_{q^2} \longrightarrow
\mathcal{F}_{q^2} \otimes \mathcal{F}_{q}\otimes \mathcal{F}_{q^2}\otimes \mathcal{F}_{q}
\end{align}
be the intertwiner for (\ref{equiv3}).
It is characterized by the following relations:
\begin{align}
&\pi_{3232}(\Delta(f))\circ \Psi = \Psi \circ \pi_{2323}(\Delta(f))
\quad (\forall f \in A_q(\mathrm{Sp}_6)),\label{pip}\\
&\Psi (|0\rangle \otimes|0\rangle \otimes|0\rangle \otimes|0\rangle) 
=|0\rangle \otimes|0\rangle \otimes|0\rangle \otimes|0\rangle, \label{pno}
\end{align}
where the latter just specifies a normalization.
We find it convenient to work with $K$ defined by
\begin{align}\label{pskp}
K = \Psi  P_{14}P_{23} : 
\; \mathcal{F}_{q^2} \otimes \mathcal{F}_{q}\otimes \mathcal{F}_{q^2}\otimes \mathcal{F}_{q}
\longrightarrow 
\mathcal{F}_{q^2} \otimes \mathcal{F}_{q}\otimes \mathcal{F}_{q^2}\otimes \mathcal{F}_{q},
\end{align}
where the composition 
$P_{14}P_{23} : x \otimes y \otimes z \otimes w \mapsto
w \otimes  z \otimes y \otimes x$ 
reverses the order of the 4-fold tensor product.
The intertwining relation (\ref{pip}) is translated into
\begin{align}
&\pi_{3232}(\Delta(f))\circ K = K \circ \pi_{3232}(\tilde{\Delta}(f))
\quad (\forall f \in A_q(\mathrm{Sp}_6)),\label{reK}
\end{align}
where $\tilde{\Delta}(f) = P_{14}P_{23} (\Delta(f))P_{14}P_{23} $, namely,
\begin{align*}
\tilde{\Delta}(t_{ij}) = \sum_{l_1,l_2,l_3} 
t_{l_3 j}\otimes t_{l_2 l_3} \otimes t_{l_1l_2} \otimes t_{i l_1}.
\end{align*}
In Theorem \ref{th:main}, 
we will see that (\ref{reK}) becomes independent of the 
the signs $(\varepsilon, \sigma, \rho)$ and the parameters
$(\alpha_i, \mu_i)$ if one switches to the 
``universal part" $\mathscr{K}$ of $K$ by (\ref{kkh}).
The resulting intertwining relations for $\mathscr{K}$ 
is listed in Appendix \ref{app:IR}.
 
Introduce the matrix elements by
\begin{align}\label{Kact}
K(|a\rangle \otimes |i\rangle \otimes |b\rangle \otimes |j\rangle) = 
\sum_{c,m,d,n} K^{c m d n}_{a \,i \,b \,j}
|c\rangle \otimes |m\rangle \otimes |d\rangle \otimes |n\rangle.
\end{align}
The normalization condition (\ref{pno}) becomes 
$K^{0000}_{0000}=1$.
Let $\mathscr{K} = (\mathscr{K}^{cmdn}_{a\, i \, b \, j})$ be the matrix 
defined via $K$ as in (\ref{kkh}).
By similar arguments to Proposition \ref{pr:Rinv} one can show
\begin{align}
&K^{-1}= K,\qquad 
\mathscr{K}^{-1} = \mathscr{K}, \label{Kinv}\\
&(q^4)_c(q^2)_m(q^4)_d(q^2)_n \,\mathscr{K}^{cmdn}_{a\, i \, b \, j}
= (q^4)_a(q^2)_i(q^4)_b(q^2)_j \,
\mathscr{K}^{a\, i \, b \, j}_{cmdn}.\label{Ksym}
\end{align}

Now we present the main formula of the paper.

\begin{theorem}\label{th:main}
The unique solution to the equation (\ref{reK}) satisfying 
$K^{0000}_{0000}=1$ has the form
\begin{align}\label{kkh}
K^{cmdn}_{a \,i \,b \,j}= \varepsilon^{m+j}
\mu_2^{2d-2b}(\rho \mu_3)^{m-i}
\mathscr{K}^{c m d n}_{a \,i \,b \,j},
\end{align}
where $\mathscr{K}^{c m d n}_{a \,i \,b \,j}$
is independent of the parameters. It is expressed as
\begin{align}
&\mathscr{K}^{c m d n}_{a \,i \,b \,j}=
\delta_{c+m+d,a+i+b}\delta_{d+n-c,b+j-a}\frac{(q^4)_a}{(q^4)_c}
\sum_{\alpha,\beta,\gamma}
\frac{(-1)^{\alpha+\gamma}}{(q^4)_{d-\beta}}q^{\phi_1}
\nonumber\\
&\qquad \times
\mathscr{K}^{a, i+b-\alpha-\beta-\gamma ,
0, j+b-\alpha-\beta-\gamma}_{c, m+d-\alpha-\beta-\gamma,0, 
n+d-\alpha-\beta-\gamma}
\left\{{b, d-\beta,  i+b-\alpha-\beta, j+b-\alpha-\beta 
\atop
\alpha, \beta, \gamma, m-\alpha, n-\alpha, 
b-\alpha-\beta, d-\beta-\gamma}\right\},
\label{bd}\\
&\phi_1 = \alpha(\alpha+2d-2\beta-1)
+(2\beta-d)(m+n+d)+\gamma(\gamma-1)-b(i+j+b), \nonumber
\end{align}
where the sum is over $\alpha, \beta, \gamma \in \Z_{\ge 0}$, which is actually finite.
The $\mathscr{K}$ in the sum is given by
\begin{align}
&\mathscr{K}^{c m 0 n}_{a \,i \,0 \,j}=
\delta_{c+m,a+i}\delta_{n-c,j-a}\sum_{\lambda}(-1)^{m+\lambda}
\frac{(q^4)_{c+\lambda}}{(q^4)_c}q^{\phi_2}
\left\{{i,j \atop \lambda, j-\lambda, m-\lambda, i-m+\lambda}\right\},\label{rest}\\
&\phi_2 = (a+c+1)(m+j-2\lambda)+m-j,\nonumber
\end{align}
where the sum is over $\lambda \in \Z_{\ge 0}$, which is actually finite.
\end{theorem}

A proof of Theorem \ref{th:main} 
is available in Appendix \ref{app:K}. 
At $d=b=0$, the formula (\ref{bd}) reduces to (\ref{Ksym}).
Note that 
\begin{equation}\label{claw}
\mathscr{K}^{cmdn}_{a\, i\, b\, j} = 0 \;\;
\text{unless}\; 
c+m+d = a+i+b\;\text{and}\; d+n-c=b+j-a.
\end{equation}
As with $\mathscr{R}$, 
this property will also be referred as the conservation law.
We have separated the result into (\ref{bd}) and 
(\ref{rest})  as the full formula 
obtained by their composition is rather bulky.

In (\ref{rest}), the quantity $\left\{ { \cdots \atop \cdots }\right\}$ is a
product of two $q^2$-binomial coefficients, therefore 
$\mathscr{K}^{c m 0 n}_{a \,i \,0 \,j}$ is a Laurent polynomial of $q$.
On the other hand, (\ref{bd})  only tells that 
$\mathscr{K}^{cmdn}_{a\, i \, b \, j}$ is a rational function of $q$ in general.
Our second main result concerns this point and 
exhibits a remarkable feature analogous  to $\mathscr{R}$ 
mentioned in Remark \ref{re:Rq0}. 

\begin{theorem}\label{th:combk}
(i) The matrix elements $\mathscr{K}^{cmdn}_{a\, i\, b\, j}$ in 
Theorem \ref{th:main} are polynomials in 
$q$ with integer coefficients. 

(ii) Set 
\begin{equation}\label{ckdef}
\mathcal{K} = (\mathcal{K}^{cmdn}_{a\,i\,b\,j}) := \mathscr{K}|_{q=0},
\end{equation}
hence 
$\mathcal{K}^{cmdn}_{a\,i\,b\,j} = \mathscr{K}^{cmdn}_{a\,i\,b\,j}|_{q=0}$.
Then it is given explicitly as
\begin{equation}\label{combk}
\begin{split}
&\mathcal{K}^{cmdn}_{a\,i\,b\,j}
=\delta_{c+m+d,a+i+b}\delta_{d+n-c,b+j-a}\,
\delta_{a,c'}\delta_{i, m'} \delta_{b,d'} \delta_{j,n'},\\
&c' = x+c+m-n,\;\;
m' = d-x+n-\min(c, d + x),\\
&d' =  \min(c, d + x),\;\;
n' = m+(d+x-c)_+,\;\;x=(d-c+(n-m)_+)_+,
\end{split}
\end{equation}
where the symbol $(y)_+$ is defined in Remark \ref{re:Rq0}.
\end{theorem}

A proof of Theorem \ref{th:combk} is outlined in Appendix \ref{app:kq0}.
We note that $\phi_1$ and $\phi_2$ in (\ref{bd}) and (\ref{rest})
can become negative in general, 
so the claim (i) implies nontrivial cancellations.
It is an interesting problem to 
construct an explicit formula of $\mathscr{K}^{cmdn}_{a\,i\,b\,j}$
in which its polynomiality is manifest.
From (\ref{bd}) and (\ref{rest}) the claim (i) can be refined to 
$\mathscr{K}^{cmdn}_{a\, i\, b\, j}
\in q^\eta \Z[q^2]$, where $\eta=0,1$ is specified by 
$\eta \equiv ij+mn\mod 2$.

From the claim (i) and (\ref{Kinv}) it follows that  
$\mathcal{K}= \mathcal{K}^{-1}$.
Thus replacing $K$ with $\mathcal{K}$ in (\ref{Kact}) 
defines a bijection.
Put in another word, $\mathcal{K}: (c,m,d,n) \mapsto (c',m',d',n')$
is a bijection on each finite set specified by the values of 
conserved quantities
$\{(c,m,d,n) \in (\Z_{\ge 0})^4\,| 
c+m+d=\text{const}, d+n-c=\text{const}\}$.
In fact the property $c',m',d',n' \in \Z_{\ge 0}$ and 
the conservation law 
$c+m+d = c'+m'+d',\, d+n-c = d'+n'-c'$ 
can easily be confirmed.
We call $\mathcal{K}$ {\em combinatorial} 3D $K$.
See \cite{KOY} for an analogous object in 2D and its application. 
We shall see another origin of the piecewise linear formula
(\ref{combk}) in Section \ref{subsec:b3dk}.

\begin{example}\label{ex:K}
The following is the list of all the nonzero $\mathscr{K}^{cmdn}_{2\,1\,1\,0}$.
\begin{align*}
\mathscr{K}^{1300}_{2110}&=q^8 (1 - q^8),\\
\mathscr{K}^{2110}_{2110} &=-q^4 (1 - q^8 + q^{14}),\\
\mathscr{K}^{2201}_{2110} &=-q^6 (1 + q^2) (1 - q^2 + q^4 - q^6 - q^{10}),\\
\mathscr{K}^{3011}_{2110} &=1 - q^8 + q^{14},\\
\mathscr{K}^{3102}_{2110} &=-q^{10} (1 - q + q^2) (1 + q + q^2),\\
\mathscr{K}^{4003}_{2110}&=q^4.
\end{align*}
Thus $\mathcal{K}^{cmdn}_{2\,1\,1\,0} 
= \delta_{c,3}\delta_{m,0}\delta_{d,1}\delta_{n,1}$ in agreement with 
(\ref{combk}).
\end{example}

\subsection{\mathversion{bold}Relation involving $R$ and $K$}

Let $w_0 \in W(\mathrm{Sp}_6)$ be the longest element of the Weyl group
and consider the two reduced expressions
\begin{equation}\label{red9}
w_0 = s_1s_2s_3s_2s_1s_2s_3s_2s_3 
= s_3s_2s_3s_2s_1s_2s_3s_2s_1,
\end{equation} 
where the order of the simple reflections are opposite in the two sides.
According to Theorem \ref{th:so}  for $\mathrm{Sp}$ case,
we have the equivalence of the 
two representations of $A_q(\mathrm{Sp}_6)$:
\begin{equation}\label{pi9}
\pi_{123212323} \simeq \pi_{323212321}.
\end{equation}

Let $P_{ij}, \Phi_{ijk}$ and $\Psi_{ijkl}$ 
be the transposition $P$ (\ref{Pt}), 
the intertwiner $\Phi$  in Section \ref{sec:R} 
and the intertwiner $\Psi$ (\ref{psdef})
that act on the tensor components specified by the indices.
As in Section \ref{ss:tea},
one can construct two 
intertwiners for (\ref{pi9}) by consulting the Coxeter relations (\ref{coxre}).
\begin{alignat}{3}
12321\underline{2323}\quad & \Psi_{6789} \qquad\qquad 
&123\underline{212}323\quad & \Phi^{-1}_{456} \nonumber\\
1232\underline{13}232 \quad & P_{56} \qquad\qquad 
&12\underline{31}2\underline{13}23\quad & P_{34}P_{67} \nonumber\\
1\underline{2323}1232 \quad & \Psi_{2345} \qquad\qquad 
&\underline{121}323123\quad & \Phi_{123} \nonumber\\
1323\underline{212}32 \quad & \Phi^{-1}_{567} \qquad\qquad 
&21\underline{2323}123\quad & \Psi_{3456} \nonumber\\
\underline{13}2\underline{31}2\underline{13}2 \quad & P_{12}P_{45}P_{78} 
\qquad\qquad &21323\underline{212}3\quad & \Phi^{-1}_{678} \nonumber\\
3\underline{121}32312 \quad & \Phi_{234} \qquad\qquad 
&2\underline{13}2\underline{31}2\underline{13}\quad & P_{23}P_{56}P_{89}\nonumber\\
321\underline{2323}12 \quad & \Psi_{4567} \qquad\qquad 
&23\underline{121}3231\quad & \Phi_{345}\nonumber\\
321323\underline{212} \quad & \Phi^{-1}_{789} \qquad\qquad 
&2321\underline{2323}1\quad & \Psi_{5678}\nonumber\\
32\underline{13}2\underline{31}21 \quad & P_{34}P_{67} \qquad\qquad 
&232\underline{13}2321\quad & P_{45}\nonumber\\
323\underline{121}321 \quad & \Phi_{456}  \qquad\qquad 
&\underline{2323}12321\quad & \Psi_{1234}\nonumber\\
323212321 \quad &  \qquad\qquad &323212321\quad & \nonumber
\end{alignat}
The underlines are assigned in the same manner as in (\ref{2ways}).
Thus we get 
\begin{equation}\label{phps}
\begin{split}
&\Phi_{456}P_{34}P_{67}\Phi^{-1}_{789}
\Psi_{4567}\Phi_{234}P_{12}P_{45}P_{78}
\Phi^{-1}_{567}\Psi_{2345}P_{56}\Psi_{6789}\\
=& \Psi_{1234}P_{45}\Psi_{5678}\Phi_{345}P_{23}P_{56}P_{89}
\Phi^{-1}_{678}\Psi_{3456}\Phi_{123}P_{34}P_{67}\Phi^{-1}_{456}.
\end{split}
\end{equation}
Substituting ($\bar{R}$ is a shorthand for $R^{-1}$ 
as in Section \ref{ss:tea})
\begin{equation*}
\Phi_{ijk} = R_{ijk} P_{jk},\quad \Phi^{-1}_{ijk} = P_{jk} \bar{R}_{ijk},\quad
\Psi_{ijkl} = K_{ijkl}P_{il}P_{jk}
\end{equation*}
into (\ref{phps}) and sending all the $P_{ij}$'s through to the right, 
we find that the products of $P_{ij}$'s
correspond to the longest element in the symmetric group $\mathfrak{S}_9$ 
on the both sides. Canceling them out, we obtain
\begin{equation}\label{rkeq1}
R_{456}\bar{R}_{984}K_{3579}R_{269}\bar{R}_{852}K_{1678}K_{1234}
=K_{1234}K_{1678}R_{258}\bar{R}_{962}K_{3579}R_{489}\bar{R}_{654}
\end{equation}
for the operators acting on $\pi_{323212321}$.
Namely, (\ref{rkeq1}) is an equality in  
$\mathrm{End}(\mathcal{F}_{q^2}\otimes \mathcal{F}_{q}\otimes 
\mathcal{F}_{q^2}\otimes \mathcal{F}_{q}\otimes \mathcal{F}_{q}\otimes \mathcal{F}_{q}\otimes 
\mathcal{F}_{q^2}\otimes \mathcal{F}_{q}\otimes \mathcal{F}_{q})$.
We call (\ref{rkeq1}) (and (\ref{rkeq2}) given below as well) 
the 3D reflection equation.

There are 42 reduced expressions for the longest element $w_0$.
We can ask if one obtains other relations than \eqref{rkeq1} from different
reduced expressions. The answer is negative as in the tetrahedron equation
\eqref{te1}. Namely, all the other relations reduce to \eqref{rkeq1},
since all 42
reduced expressions of $w_0$ essentially appear in either course from
123212323 to 323212321 to get \eqref{phps}.

Let us discuss the physical interpretation of the construction here.
As mentioned in the introduction, the relevant system is 
the factorized scattering of strings in 3D \cite{Zam80,Zam81}
under the presence of boundary reflections \cite{IK}.
Our 3D reflection equation (\ref{rkeq1})
is a constant version of the tetrahedron reflection equation \cite{IK}.
This can be seen by relabeling the space indices in (\ref{rkeq1})
as $(1,2,3,4,5,6,7,8,9)\rightarrow 
(\overline{x}, 6,\overline{y},5,4,3,\overline{z},2,1)$.
The resulting equation essentially coincides with \cite[eq.(17)]{IK}.
The spaces $2,4,5,6,8,9$ in (\ref{rkeq1}) correspond to $\mathcal{F}_q$ 
and are attached with strings.
The other spaces labeled by $1,3$ and $7$ are $\mathcal{F}_{q^2}$, 
hence touched upon only by $K$'s.
They represent a physical degree of freedom living on 
the boundary which are subject to quantum transitions 
when reflecting back strings.

The fundamental representations $\pi_1, \pi_2, \pi_3$ of 
$A_q(\mathrm{Sp}_6)$ in this paper correspond to
the generators $A_1,\ldots, A_6, B^{\overline{u}}$ in the 3D 
analogue of the Zamolodchikov algebra \cite{IK}  as 
\begin{equation*}
\pi_1 \rightarrow A_\bullet,\;\;\pi_2 \rightarrow A_\bullet,\;\;
\pi_3 \rightarrow B^\bullet.
\end{equation*}
In fact, consider for instance the LHS of (\ref{pi9}).
We let the ninth order tensor product correspond to 
a word of $A_i$'s and $B^{\overline{u}}$ by the above rule as
\begin{equation*}
\pi_1 \otimes \pi_2 \otimes  \pi_3\otimes 
\pi_2\otimes \pi_1\otimes \pi_2\otimes 
\pi_3\otimes \pi_2\otimes \pi_3 \rightarrow 
A_1A_2B^{\overline{z}}A_3A_4A_5B^{\overline{y}}A_6B^{\overline{x}},
\end{equation*}
where indices of $A$ ($B$) are assigned consecutively (inverse alphabetically).
The RHS reproduces the element in \cite[(c) p434]{IK}.
Similarly the intertwining relations of 
$R$ (\ref{reRA}) and $K$ (\ref{reK})
correspond to the 3D Zamolodchikov algebra \cite[eq.(4)]{IK} and 
the boundary reflection algebra \cite[eq.(9)]{IK}, respectively.
The reordering process of 
$A_1A_2B^{\overline{z}}A_3A_4A_5B^{\overline{y}}A_6B^{\overline{x}}$
is translated into the composition of intertwiners $K$ and $R$
as demonstrated after (\ref{pi9}).
Physically $A_i$ represents a straight string moving in 3D or the world sheet 
generated by it, and $B^{\overline{u}}$ denotes a reflection by the boundary.
The three $A_i$'s among the six correspond to 
the strings heading toward the boundary
and the other three to those going off the boundary after the reflections 
represented by the three $B^{\overline{u}}$'s.

Now we turn to a combinatorial aspect of the 3D reflection equation (\ref{rkeq1}) .
When all the parameters are removed, it becomes
\begin{equation}\label{rkeq2}
\mathscr{R}_{456}\mathscr{R}_{489}
\mathscr{K}_{3579}\mathscr{R}_{269}\mathscr{R}_{258}
\mathscr{K}_{1678}\mathscr{K}_{1234}
=\mathscr{K}_{1234}
\mathscr{K}_{1678}
\mathscr{R}_{258}\mathscr{R}_{269}
\mathscr{K}_{3579}\mathscr{R}_{489}\mathscr{R}_{456},
\end{equation}
where we have applied $\mathscr{R}^{-1} = \mathscr{R}$ (\ref{RAinv})
and $\mathscr{R}_{ijk}= \mathscr{R}_{kji}$ 
due to (\ref{swapR}).
This is an identity between polynomials of $q$.
Setting $q=0$ further and invoking  
Theorem \ref{th:combk} and Remark \ref{re:Rq0}, we find that 
the combinatorial 3D $R$ and the combinatorial 3D $K$ still satisfy 
\begin{equation}\label{rkeq3}
\mathcal{R}_{456}\mathcal{R}_{489}
\mathcal{K}_{3579}\mathcal{R}_{269}\mathcal{R}_{258}
\mathcal{K}_{1678}\mathcal{K}_{1234}
=\mathcal{K}_{1234}
\mathcal{K}_{1678}
\mathcal{R}_{258}\mathcal{R}_{269}
\mathcal{K}_{3579}\mathcal{R}_{489}\mathcal{R}_{456}.
\end{equation}
This is an identity of bijections between finite subsets of 
$(\Z_{\ge 0})^9$, which may be called 
the combinatorial 3D reflection equation.
It is an interesting problem whether it leads  
to a 3D generalization of the result like \cite{KOY}.

\begin{example}\label{ex:cref}
To demonstrate (\ref{rkeq3}),
we employ the same convention as in 
Example \ref{ex:cr}.
Then the monomial, say, $|211034212\rangle$ is 
transformed as in Figure \ref{fig:cref}.
The first SW arrow by $\mathcal{K}_{1234}$ is due to 
Example \ref{ex:K}.
 
\begin{figure}[h]
\begin{picture}(100,225)(10,-90)
\put(50,120){$|211034212\rangle$}

\put(7,110){$\mathcal{K}_{1234}\; \swarrow$}
\put(105,110){$\searrow \,\mathcal{R}_{456}$}

\put(0,90){$|301134212\rangle$} \put(110,90){$|211307212\rangle$} 

\put(-7,75){$\mathcal{K}_{1678} \downarrow$}
\put(130,75){$\downarrow \mathcal{R}_{489}$}

\put(0,60){$|601131242\rangle$} \put(110,60){$|211207221\rangle$} 

\put(-1,45){$\mathcal{R}_{258} \downarrow$}
\put(130,45){$\downarrow \mathcal{K}_{3579}$}

\put(0,30){$|631101272\rangle$} \put(110,30){$|212207123\rangle$} 

\put(-1,15){$\mathcal{R}_{269} \downarrow$}
\put(130,15){$\downarrow \mathcal{R}_{269}$}

\put(0,0){$|621102271\rangle$} \put(110,0){$|272201129\rangle$} 

\put(-7,-15){$\mathcal{K}_{3579} \downarrow$}
\put(130,-15){$\downarrow \mathcal{R}_{258}$}

\put(0,-30){$|622102173\rangle$} \put(110,-30){$|252221109\rangle$} 

\put(-1,-45){$\mathcal{R}_{489} \downarrow$}
\put(130,-45){$\downarrow \mathcal{K}_{1678}$}

\put(0,-60){$|622702119\rangle$} \put(110,-60){$|352220119\rangle$} 

\put(7,-75){$\mathcal{R}_{456}\searrow$}
\put(105,-75){$\swarrow \,\mathcal{K}_{1234}$}

\put(50,-90){$|622520119\rangle$}
\end{picture}    
\caption{An example of the 
3D reflection equation (\ref{rkeq3}) 
for the combinatorial 3D $R$ and $K$.}
\label{fig:cref}
\end{figure}
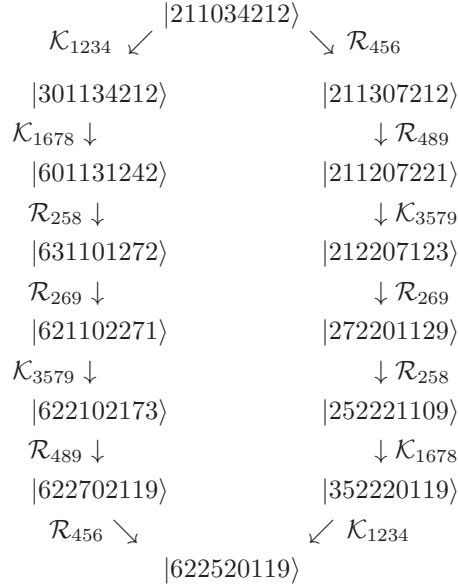
\end{example}

\subsection{\mathversion{bold}Birational 3D $K$}\label{subsec:b3dk}
Introduce the upper triangular matrices
\begin{alignat*}{2}
X_1(z) &= \begin{pmatrix}
1 & z & 0 & 0\\
   & 1 & 0 & 0\\
   &    & 1 & -z\\
   &    &    & 1
\end{pmatrix},\qquad &
X_2(z) &= \begin{pmatrix}
1 & 0 & 0 & 0\\
   & 1 & 2z & 0\\
   &    & 1 & 0\\
   &    &    & 1
\end{pmatrix},\\
Y_1(z) &= \begin{pmatrix}
1 & z & 0 & 0 & 0\\
   & 1 & 0 & 0 & 0\\
   &    & 1 & 0 & 0\\
   &    &    & 1 & -z\\
   &    &    &    & 1
\end{pmatrix},&
Y_2(z) &= \begin{pmatrix}
1 & 0 & 0 & 0 & 0\\
   & 1 & z & -z^2/2 & 0\\
   &    & 1 & -z & 0\\
   &    &    & 1 & 0\\
   &    &    &    & 1
\end{pmatrix},
\end{alignat*}
where blanks signify  $0$ and $z$ is a parameter.
The matrix $X_i(z)$ is a generator of the unipotent subgroup of 
$\mathrm{Sp}_4$.
Similarly  $Y_i(z)$ is the one for $\mathrm{SO}_5$.
They are associated with two realizations of the Lie group
corresponding to $\mathrm{Sp}_4 \simeq \mathrm{SO}_5$.
The matrices $X_2(z)$ and $Y_1(z)$ correspond 
to the long simple root, so the role of indices $1$ and $2$ are 
interchanged in the two pictures.
They satisfy $X_i(z)^{-1} = X_i(-z)$ and 
$Y_i(z)^{-1} = Y_i(-z)$.
By a direct calculation one can establish
\begin{theorem}\label{bcXY}
Given indeterminates $(a,b,c,d)$, 
each of the two matrix equations 
\begin{align}
X_2(a)X_1(b)X_2(c)X_1(d) &= 
X_1(\tilde{a})X_2(\tilde{b})X_1(\tilde{c})X_2(\tilde{d}),
\label{Xeq}\\
Y_1(a)Y_2(b)Y_1(c)Y_2(d) &= 
Y_2(\tilde{a})Y_1(\tilde{b})Y_2(\tilde{c})Y_1(\tilde{d})
\label{Yeq}
\end{align}
for $(\tilde{a},\tilde{b},\tilde{c},\tilde{d})$ 
has the unique solution
\begin{equation}\label{abcdt}
\begin{split}
&\tilde{a} = \frac{bcd}{A},\quad
\tilde{b} = \frac{A^2}{B},\quad
\tilde{c} = \frac{B}{A},\quad
\tilde{d} = \frac{ab^2c}{B},\\
&A = ab + ad + cd,\quad
B = ab^2+2abd + ad^2+cd^2.
\end{split}
\end{equation}
\end{theorem}

Define a map 
\begin{equation}\label{bKdef}
{\bf K}: (d,c,b,a) \mapsto (\tilde{a},\tilde{b},\tilde{c},\tilde{d})
\end{equation}
in terms of (\ref{abcdt}),
where the reason for not $(a,b,c,d)$ but $(d,c,b,a)$ is 
to fit $K = \Psi P_{14}P_{23}$ in (\ref{pskp}).
It is easy to see ${\bf K}^{-1} = {\bf K}$, hence ${\bf K}$ is birational.
We call ${\bf K}$ the {\em birational} 3D $K$.
The intertwining relation (\ref{pip}) is a quantization 
of (\ref{Xeq}).
(For $\mathrm{Sp}_4$, $\pi_{3232}$ therein 
should read $\pi_{2121}$.)
It satisfies the formally identical equation with (\ref{rkeq2}):
\begin{equation}\label{rkeq4}
{\bf R}_{456}{\bf R}_{489}
{\bf K}_{3579}{\bf R}_{269}{\bf R}_{258}
{\bf K}_{1678}{\bf K}_{1234}
={\bf K}_{1234}
{\bf K}_{1678}
{\bf R}_{258}{\bf R}_{269}
{\bf K}_{3579}{\bf R}_{489}{\bf R}_{456}.
\end{equation}
This is an equality of the birational maps on 9 variables which can be
directly checked.
Alternatively it can also be derived following the argument
similar to (\ref{Gpro})--(\ref{bth}).

The birational 3D $K$ tends to the combinatorial 3D $K$ via 
the ultradiscretization.
In fact by the tropical variable change   
$\alpha\beta \rightarrow \alpha+ \beta$
and $\alpha+\beta \rightarrow \min(\alpha,\beta)$,
the formulas (\ref{abcdt})--(\ref{bKdef}) 
exactly reproduce the piecewise-linear map 
${\mathcal K}: (c,m,d,n) \mapsto (c',m',d',n')$ 
in (\ref{combk}).
Thus we have realized the triad of 3D $K$'s in Table 1 
all satisfying the 3D reflection equations
(\ref{rkeq2}), (\ref{rkeq3}) and (\ref{rkeq4}).

\section{Type $B$ and $F_4$ cases}\label{sec:BF}

From the construction in the preceding sections for 
algebras of type $A$ and $C$,
it is quite possible to infer the situation in type 
$B$ and $F_4$.
In this section we discuss them 
without a proof or concrete realizations
of the representations.
We shall only be concerned with
the parameter-free part of the intertwiners like 
$\mathscr{R}$ and $\mathscr{K}$ which are polynomials of $q$ only.
The minimal and generic situation for the $B$ series 
takes place for $B_3$.
We list the relevant 
Dynkin diagrams in Figure \ref{fig:dynkin}.

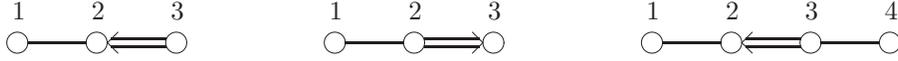
\begin{figure}[h]
\begin{picture}(300,23)(0,3)

\put(-2,19){1}\put(28,19){2}\put(58,19){3}
\put(0,10){\circle{8}}\put(30,10){\circle{8}}\put(60,10){\circle{8}}
\put(4,10){\line(1,0){22}}
\put(34,10){\line(1,1){4}}\put(34,10){\line(1,-1){4}}
\put(56,8.5){\line(-1,0){20.5}}\put(56,11.5){\line(-1,0){20.5}}

\put(120,0){
\put(-2,19){1}\put(28,19){2}\put(58,19){3}
\put(0,10){\circle{8}}\put(30,10){\circle{8}}\put(60,10){\circle{8}}
\put(4,10){\line(1,0){22}}
\put(56,10){\line(-1,1){4}}\put(56,10){\line(-1,-1){4}}
\put(34,8.5){\line(1,0){20.5}}\put(34,11.5){\line(1,0){20.5}}}

\put(240,0){
\put(-2,19){1}\put(28,19){2}\put(58,19){3}\put(88,19){4}
\put(0,10){\circle{8}}\put(30,10){\circle{8}}\put(60,10){\circle{8}}\put(90,10){\circle{8}}
\put(4,10){\line(1,0){22}}\put(64,10){\line(1,0){22}}
\put(34,10){\line(1,1){4}}\put(34,10){\line(1,-1){4}}
\put(56,8.5){\line(-1,0){20.5}}\put(56,11.5){\line(-1,0){20.5}}
}

\end{picture}
\caption{Dynkin diagrams of $C_3$ (left), $B_3$ (center) and $F_4$ (right).
Enumeration of vertices for $F_4$ agrees with 
\cite{Bourbaki} which is opposite to \cite{Kac}.  
}
\label{fig:dynkin}
\end{figure}

\smallskip
{\em Type $B$ case}.
The Weyl group $W(B_3)$ is isomorphic to $W(C_3)$.
Thus we should have the equivalence (\ref{equiv1})--(\ref{equiv3})
for the irreducible representations 
$\pi_i = \pi^B_i (i=1,2,3)$ of the quantized algebra of
functions $A_q(\mathrm{SO}_7)$,
where  
$\mathrm{SO}_7$ is the Lie group corresponding to $B_3$.
Note that locally in the Dynkin diagrams,
the role of the indices $2$ and $3$ are interchanged 
between $B_3$ and $C_3$.
From this fact and (\ref{reK}),
the intertwiner 
$\mathscr{K}^B$ satisfying 
$(\pi^B_{3232} \Delta)\circ \mathscr{K}^B = 
\mathscr{K}^B \circ (\pi^B_{3232} \tilde{\Delta})$
should be obtained from the type $C$ case $\mathscr{K}$ as
\begin{equation}\label{KB}
\mathscr{K}^B_{1234} = P_{14}P_{23}\mathscr{K}_{1234}P_{23}P_{14}
= \mathscr{K}_{4321} 
\in 
\mathrm{End}(\mathcal{F}_q \otimes \mathcal{F}_{q^2} 
\otimes \mathcal{F}_q \otimes \mathcal{F}_{q^2}).
\end{equation}
Here the second equality just means that its RHS  
is the standard notation for the middle object
acting on the tensor components labeled with $1,2,3$ and $4$.

Now we proceed to the intertwiner $\mathscr{R}^B$ satisfying 
$(\pi^B_{212} \Delta)\circ \mathscr{R}^B = 
 \mathscr{R}^B\circ (\pi^B_{121} \tilde{\Delta})$.
Since the segment of the Dynkin diagrams between the vertices 
$1$ and $2$ are the same for $B_3$ and $C_3$, we expect that 
$\mathscr{R}^B$ is obtained 
from the type $C$ case (hence $A$ case) $\mathscr{R}$ as
\begin{equation}\label{RB}
\mathscr{R}^B = \mathscr{S}:=\mathscr{R}|_{q\rightarrow q^2}
\in 
\mathrm{End}(\mathcal{F}_{q^2}\otimes \mathcal{F}_{q^2}
\otimes \mathcal{F}_{q^2}).
\end{equation}
Here the replacement $q \rightarrow q^2$ reflects the 
squared length of the simple roots 
attached to the vertices $1$ and $2$ compared with $3$.
It is an opposite of (\ref{qi}).

To summarize so far, we conjecture that 
the intertwiners $\mathscr{K}^B$ and $\mathscr{R}^B$ for 
$A_q(\mathrm{SO}_7)$ are obtained from the corresponding objects 
in $A_q(\mathrm{Sp}_6)$ by the simple prescriptions (\ref{KB}) and (\ref{RB}).

Now we consider the 3D reflection equations.
Due to $W(B_3) \simeq W(C_3)$, 
the intertwiners $\mathscr{K}^B$ and $\mathscr{R}^B$
should fulfill exactly the same relation as (\ref{rkeq2}).
In other words, (\ref{rkeq2}) should survive 
under the replacement
$(\mathscr{R}, \mathscr{K}) \rightarrow 
(\mathscr{R}^B, \mathscr{K}^B)$.
Thus we conjecture that 
\begin{equation}\label{b3eq}
\mathscr{S}_{4 5 6} \mathscr{S}_{4 8 9} \mathscr{K}_{9 7 5 3} \mathscr{S}_{2 6 9} 
  \mathscr{S}_{2 5 8} \mathscr{K}_{8 7 6 1} \mathscr{K}_{4 3 2 1}
=
\mathscr{K}_{4 3 2 1} \mathscr{K}_{8 7 6 1} \mathscr{S}_{2 5 8} 
\mathscr{S}_{2 6 9} 
   \mathscr{K}_{9 7 5 3} \mathscr{S}_{4 8 9} \mathscr{S}_{4 5 6}
\end{equation}
holds in 
$\mathrm{End}(\mathcal{F}_{q}\otimes \mathcal{F}_{q^2}
\otimes \mathcal{F}_{q}\otimes \mathcal{F}_{q^2}
\otimes \mathcal{F}_{q^2}\otimes \mathcal{F}_{q^2}
\otimes \mathcal{F}_{q}\otimes \mathcal{F}_{q^2}
\otimes \mathcal{F}_{q^2})$.
This is a yet independent relation from (\ref{rkeq2}) to be called 
the {\em 3D reflection equation of type $B$}.
The previous one (\ref{rkeq2}) is of type $C$ in this context.
We have checked (\ref{b3eq})  by computer for several examples.
For instance when the both sides act on the monomial 
$|1 1 2 1 1 1 1 1 1\rangle$ specified by the 
occupation numbers of $q$-oscillators,
they generate the same vector consisting of $1410$ monomials.

\smallskip
{\em $F_4$ case}. We let $A_q(\mathrm{F}_4)$ 
denote the quantized algebra of 
functions on the Lie group corresponding to $F_4$.
Let $\pi_i$ be its irreducible representation 
attached to the vertex $i$ of the Dynkin diagram in Figure \ref{fig:dynkin}.
We expect that it is realized in terms of the $q$-oscillators as
\begin{equation*}
\pi_i :   A_q(\mathrm{F}_4) 
\rightarrow \mathrm{End}(\mathcal{F}_{q_i})\quad
\text{with}\;\;
(q_1, q_2, q_3, q_4) = (q,q, q^2, q^2).
\end{equation*}
The $q_i$ here is a natural prolongation of (\ref{qi})
reflecting the squared length of the simple roots.
The Coxeter relations for the simple reflections 
$s_i \in W(\mathrm{F}_4)$ are
given by
\begin{equation}\label{f4cox}
s_1s_2s_1 = s_2s_1s_2,\;\;
s_2s_3s_2s_3 = s_3s_2s_3s_2,\;\;
s_3s_4s_3 = s_4s_3s_4
\end{equation}
in addition to the `trivial' ones  
$s_i^2=1$ and $s_is_j = s_js_i$ for $|i-j|>1$.
Thus one should have the equivalence between the corresponding 
tensor products of $\pi_i$'s \cite{So1, So2}.
Introduce the 
three kinds of intertwiners 
$\mathscr{R}^F, \mathscr{K}^F$ and $\mathscr{S}^F$ 
characterized up to normalization by
\begin{equation}\label{RKSF}
\begin{split}
(\pi_{212}\Delta) \circ \mathscr{R}^F 
&= \mathscr{R}^F \circ (\pi_{121}\tilde{\Delta}),\\
(\pi_{3232}\Delta) \circ \mathscr{K}^F
&=  \mathscr{K}^F \circ (\pi_{3232}\tilde{\Delta}),\\
(\pi_{434}\Delta) \circ \mathscr{S}^F 
&= \mathscr{S}^F \circ (\pi_{343}\tilde{\Delta}).
\end{split}
\end{equation}
From the Dynkin diagrams in Figure \ref{fig:dynkin}
and the discussion for type $B$, it is natural to conjecture that 
they are simply related to the preceding ones as
\begin{equation}\label{RKS}
\mathscr{R}^F = \mathscr{R},\quad
\mathscr{K}^F = \mathscr{K},\quad
\mathscr{S}^F = \mathscr{S}.
\end{equation}
The RHSs have been encountered first for type 
$A$, $C$ and $B$ and described explicitly in  
(\ref{rA2}), Theorem \ref{th:main} and (\ref{RB}), respectively.

What about the identities analogous to the 3D reflection equations?
We can follow the same argument as before to derive them
from a reduced word for the longest element 
$w_0 \in W(\mathrm{F}_4)$ by reversing it 
in two ways.
The length of $w_0$ is 24 and there are 
2144892 reduced expressions for it.
We have picked up 
\begin{equation}\label{w0f4}
s_4 s_3 s_4 s_2 s_3 s_4 s_2 s_3 
s_2 s_1 s_2 s_3 s_4 s_2 s_3 s_1 s_2 s_3 s_4 s_1 s_2 s_3 s_2 s_1
\end{equation}
and reversed the ordering via the Coxeter relations
(\ref{f4cox}).
By representing the procedure in terms of the intertwiners (\ref{RKS}) 
we find that the consistency is expressed as
\begin{equation}\label{eqf4}
\begin{split}
&\mathscr{S}_{14, 15, 16} \mathscr{S}_{9, 11, 16} \mathscr{K}_{16, 10, 8, 7} 
\mathscr{K}_{9, 13, 15, 17} \mathscr{S}_{4, 5, 16} 
\mathscr{R}_{7, 12, 17} \mathscr{S}_{1, 2, 16} \mathscr{R}_{6, 10, 17} 
\mathscr{S}_{9, 14, 18} \mathscr{K}_{1, 3, 5, 17}\\
\times \,&\mathscr{S}_{11, 15, 18} \mathscr{K}_{18, 12, 8, 6} \mathscr{S}_{1, 4, 18}
\mathscr{S}_{1, 8, 15} \mathscr{R}_{7, 13, 19}
\mathscr{K}_{1, 6, 11,19} \mathscr{K}_{4, 12, 15, 19} \mathscr{R}_{3, 10, 19}
 \mathscr{S}_{4, 8, 11} \mathscr{K}_{1, 7,14, 20} \\
\times \,&\mathscr{S}_{2, 5, 18} \mathscr{R}_{6, 13, 20} \mathscr{R}_{3, 12, 20} 
\mathscr{S}_{1, 9, 21} \mathscr{K}_{2, 10, 15, 20}
\mathscr{S}_{4, 14, 21} \mathscr{K}_{21, 13, 8, 3} \mathscr{S}_{2,11, 21} 
\mathscr{S}_{2, 8, 14} \mathscr{R}_{6, 7, 22} \\
\times \,&\mathscr{K}_{2, 3, 4, 22} \mathscr{S}_{5, 15, 21} \mathscr{K}_{11, 13, 14, 22} 
\mathscr{R}_{10, 12, 22} \mathscr{K}_{2, 6, 9, 23} 
\mathscr{R}_{3, 7, 23} \mathscr{R}_{19, 20, 22} \mathscr{K}_{16, 17, 18, 22}
 \mathscr{R}_{10, 13, 23} \mathscr{K}_{5, 12, 14, 23} \\
\times \,&\mathscr{R}_{3, 6, 24} \mathscr{K}_{16, 19, 21, 23} \mathscr{K}_{4, 7, 9, 24} 
\mathscr{R}_{17, 20, 23} \mathscr{K}_{5, 10, 11, 24}
\mathscr{R}_{12, 13, 24} \mathscr{R}_{17, 19, 24} \mathscr{K}_{18, 20, 21, 24} 
\mathscr{S}_{5, 8, 9} \mathscr{R}_{22, 23, 24}\\
= & \; \text{product in reverse order}.
\end{split}
\end{equation}   
Here we have already applied the properties 
$\mathscr{R}^{-1} = \mathscr{R}$ (\ref{RAinv}) and 
$\mathscr{R}_{i,j,k} = \mathscr{R}_{k,j,i}$ (\ref{swapR}).
Each side consists of $16$ $\mathscr{R}$'s, $16$ $\mathscr{S}$'s and 
$18$ $\mathscr{K}$'s 
(3 of them having decreasing order of indices) 
amounting to 50 factors in total.
So things get monstrous somewhat
as is usual for exceptional Lie algebras,
and a physical interpretation seems formidable for 
(\ref{eqf4}).
However it should be emphasized that its validity is a 
{\em corollary} of \cite{So1, So2} provided that 
$\mathscr{R}, \mathscr{K}$ and $\mathscr{S}$ 
are really the intertwiners
$\mathscr{R}^F, \mathscr{K}^F$ and $\mathscr{S}^F$
for $A_q(\mathrm{F}_4)$
characterized by (\ref{RKSF}).
This last point, i.e.
(\ref{RKS}) is the only conjectural aspect in  (\ref{eqf4}).
Again we have confirmed it by computer in several examples, 
which are limited however considerably to small ones.
For instance when the both sides act on the monomial 
$|1 1 1 1 0 1 1 0 1 0 1 0 1 0 1 1 0 2 1 1 0 1 0 1\rangle$,
they both generate the same vector consisting of $533$
monomials at least mod $q^6\Z[q]$.

Besides the conjecture, 
we close with two questions which are yet to be answered.
First, can any other consistency relation
involving $\mathscr{R}, \mathscr{K}$ and $\mathscr{S}$, 
say those stemming from 
other reduced expressions, 
be attributed to (\ref{eqf4})?
For $A_q(\mathrm{Sp}_6)$ the analogous question had a 
positive answer.
See the remark after (\ref{rkeq1}). 
Second, can (\ref{eqf4}) be attributed to a composition of the 
tetrahedron equations of $\mathscr{R}, \mathscr{S}$ and 
the 3D reflection equations of type $C$ (\ref{rkeq2})  and type $B$ (\ref{b3eq})?
We hope to report on these issues together with 
$D_n, E_{6,7,8}$ and $G_2$ cases 
in a separate publication.

\appendix
\section{Intertwining relations for $\mathscr{K}$}\label{app:IR}

Let $\langle r s \rangle$ be the 
intertwining relation for $\mathscr{K}$ obtained by 
substituting (\ref{kkh}) into (\ref{reK}) with the choice 
$f = t_{rs}$.
They are independent of the parameters other than $q$.
The relation $\langle r s\rangle$ holds trivially as $0=0$ 
unless $2 \le r,s \le 5$.
The nontrivial cases are given as follows.
\begin{alignat*}{2}
&\langle 2 2 \rangle: & \quad &
[ \ichi \!\otimes\! \am \!\otimes\! \ichi \!\otimes\! \am
- q \ichi \!\otimes\! \ok \!\otimes\! \Am \!\otimes\! \ok, \,\mathscr{K}] = 0,\\
&\langle 2 3 \rangle: & \quad &
(\ichi \!\otimes\! \am \!\otimes\! \ichi \!\otimes\! \ok
+ \ichi \!\otimes\! \ok \!\otimes\! \Am \!\otimes\! \ap)\mathscr{K}\\
&&& =
\mathscr{K}(\Am \!\otimes\! \ap \!\otimes\! \Am \!\otimes\! \ok+
\Am \!\otimes\! \ok \!\otimes\! \ichi \!\otimes\! \am -q^2
\OK \!\otimes\! \am \!\otimes\! \OK \!\otimes\! \ok),\\
&\langle 2 4 \rangle: & \quad &
(\ichi \!\otimes\! \ok \!\otimes\! \OK \!\otimes\! \am) \mathscr{K} = 
\mathscr{K}(
\Ap \!\otimes\! \am \!\otimes\! \OK \!\otimes\! \ok + 
\OK \!\otimes\! \ap \!\otimes\! \Am \!\otimes\! \ok +
\OK \!\otimes\! \ok \!\otimes\! \ichi \!\otimes\! \am),\\
&\langle 2 5 \rangle: & \quad &
[\ichi \!\otimes\! \ok \!\otimes\! \OK \!\otimes\! \ok, \mathscr{K}] = 0,\\
&\langle 3 2 \rangle: & \quad &
(\Am \!\otimes\! \ap \!\otimes\! \Am \!\otimes\! \ok
+\Am \!\otimes\! \ok \!\otimes\! \ichi \!\otimes\! \am -q^2
\OK \!\otimes\! \am \!\otimes\! \OK \!\otimes\! \ok)\mathscr{K}\\
&&& =\mathscr{K}( \ichi \!\otimes\! \am \!\otimes\! \ichi \!\otimes\! \ok +
\ichi \!\otimes\! \ok \!\otimes\! \Am \!\otimes\! \ap),\\
&\langle 3 3 \rangle: & \quad &
[\Am \!\otimes\! \ap \!\otimes\! \Am \!\otimes\! \ap -q
\Am \!\otimes\! \ok \!\otimes\! \ichi \!\otimes\! \ok -q^2
\OK \!\otimes\! \am \!\otimes\! \OK \!\otimes\! \ap, \,\mathscr{K}] = 0,\\
&\langle 3 4 \rangle: & \quad &
(\Am \!\otimes\! \ap \!\otimes\! \OK \!\otimes\! \am + 
\OK \!\otimes\! \am \!\otimes\! \Ap \!\otimes\! \am - q
\OK \!\otimes\! \ok \!\otimes\! \ichi \!\otimes\! \ok)\mathscr{K}\\
&&&=
\mathscr{K}(\Ap \!\otimes\! \am \!\otimes\! \OK \!\otimes\! \ap+
\OK \!\otimes\! \ap \!\otimes\! \Am \!\otimes\! \ap -q
\OK \!\otimes\! \ok \!\otimes\! \ichi \!\otimes\! \ok),\\
&\langle 3 5 \rangle: & \quad &
(\Am \!\otimes\! \ap \!\otimes\! \OK \!\otimes\! \ok+
\OK \!\otimes\! \am \!\otimes\! \Ap \!\otimes\! \ok+
\OK \!\otimes\! \ok \!\otimes\! \ichi \!\otimes\! \ap )\mathscr{K}=
\mathscr{K}(
\ichi \!\otimes\! \ok \!\otimes\! \OK \!\otimes\! \ap),\\
&\langle 4 2 \rangle: & \quad &
(\Ap \!\otimes\!\am \!\otimes\! \OK \!\otimes\! \ok + 
\OK \!\otimes\! \ap \!\otimes\! \Am \!\otimes\! \ok +
\OK \!\otimes\! \ok \!\otimes\! \ichi \!\otimes\! \am)\mathscr{K}
=
\mathscr{K}(\ichi \!\otimes\! \ok \!\otimes\! \OK \!\otimes\! \am),\\
&\langle 4 3 \rangle: & \quad &
(\Ap \!\otimes\! \am \!\otimes\! \OK \!\otimes\! \ap+
\OK \!\otimes\! \ap \!\otimes\! \Am \!\otimes\! \ap - q
\OK \!\otimes\! \ok \!\otimes\! \ichi \!\otimes\! \ok)\mathscr{K}\\
&&& =
\mathscr{K}(\Am \!\otimes\! \ap \!\otimes\! \OK \!\otimes\! \am +
\OK \!\otimes\! \am \!\otimes\! \Ap \!\otimes\! \am -q
\OK \!\otimes\! \ok \!\otimes\! \ichi \!\otimes\! \ok),\\
&\langle 4 4 \rangle: & \quad &
[\Ap \!\otimes\! \am \!\otimes\! \Ap \!\otimes\! \am -q 
\Ap \!\otimes\! \ok \!\otimes\! \ichi \!\otimes\! \ok -q^2
\OK \!\otimes\! \ap \!\otimes\! \OK \!\otimes\! \am,\, \mathscr{K}]=0,\\
&\langle 4 5 \rangle: & \quad &
(\Ap \!\otimes\! \am \!\otimes\! \Ap \!\otimes\! \ok +
\Ap \!\otimes\! \ok \!\otimes\! \ichi \!\otimes\! \ap -q^2
\OK \!\otimes\! \ap \!\otimes\! \OK \!\otimes\! \ok)\mathscr{K}\\
&&&= \mathscr{K}(\ichi \!\otimes\! \ap \!\otimes\! \ichi \!\otimes\! \ok +
\ichi \!\otimes\! \ok \!\otimes\! \Ap \!\otimes\! \am),\\
&\langle 5 2 \rangle: & \quad &
[\ichi \!\otimes\! \ok \!\otimes\! \OK \!\otimes\! \ok,\, \mathscr{K}] = 0
\quad 
(\text{same as $\langle 2 5 \rangle$}),\\
&\langle 5 3 \rangle: & \quad &
(\ichi \!\otimes\! \ok \!\otimes\! \OK \!\otimes\! \ap)\mathscr{K} = 
\mathscr{K}(\Am \!\otimes\! \ap \!\otimes\! \OK \!\otimes\! \ok +
\OK \!\otimes\! \am \!\otimes\! \Ap \!\otimes\! \ok+
\OK \!\otimes\! \ok \!\otimes\! \ichi \!\otimes\! \ap),\\
&\langle 5 4 \rangle: & \quad &
(\ichi \!\otimes\! \ap \!\otimes\! \ichi \!\otimes\! \ok + 
\ichi \!\otimes\! \ok \!\otimes\! \Ap \!\otimes\! \am)\mathscr{K}\\
&&&=\mathscr{K}(\Ap \!\otimes\! \am \!\otimes\! \Ap \!\otimes\! \ok+
\Ap \!\otimes\! \ok \!\otimes\! \ichi \!\otimes\! \ap - q^2
\OK \!\otimes\! \ap \!\otimes\! \OK \!\otimes\! \ok),\\
&\langle 5 5 \rangle: & \quad &
[ \ichi \!\otimes\! \ap \!\otimes\! \ichi \!\otimes\! \ap -q
\ichi \!\otimes\! \ok \!\otimes\! \Ap \!\otimes\! \ok,\, \mathscr{K}] = 0.
\end{alignat*}
For example (\ref{eq55})
is obtained by taking the matrix element of $\langle 5 5 \rangle$ 
for the transition 
$|a\rangle \otimes |i\rangle \otimes |b-1\rangle \otimes |j\rangle
\rightarrow 
|c\rangle \otimes |m\rangle \otimes |d\rangle \otimes |n\rangle$.

The equations $\langle r s \rangle$ and  $\langle s r \rangle$
are transformed into each other by 
$\mathscr{K} \leftrightarrow \mathscr{K}^{-1}$.
Combining (\ref{Kinv}) and (\ref{Ksym}) with 
the argument similar to Proposition \ref{pr:Rinv}, one can also show that 
$\langle r s \rangle$ and  
$\langle s' r' \rangle$ are transformed into each other 
via the simultaneous interchange 
$(\ap, \Ap, \mathscr{K}) \leftrightarrow 
(\am, \Am, \mathscr{K}^{-1})$, where $r'=7-r$.

\section{Proof of Theorem \ref{th:main}}\label{app:K}

Let $\langle r s\rangle$ be the intertwining relation for 
the matrix elements $\mathscr{K}^{cmdn}_{a\, i\, b\, j}$
as explained in Appendix \ref{app:IR}.
The equations
$\langle 2 5\rangle$ and $\langle 5 2\rangle$ give 
$(q^{2b+i+j}-q^{2d+m+n})\mathscr{K}^{cmdn}_{a\,i\,b\,j}=0$.
Note also that 
$\mathscr{K}^{cmdn}_{a\,i\,b\,j}=0$ unless $a,b,c,d,i,j,m,n\ge 0$.
Applying these properties 
and the normalization $\mathscr{K}^{0000}_{0000}=1$
to the rest of $\langle r s\rangle$, 
one can deduce the conservation law  
implied by the factor 
$\delta_{c+m+d,a+i+b}\delta_{d+n-c,b+j-a}$ in (\ref{bd}).

\smallskip
First we reduce $\mathscr{K}^{cmdn}_{a\,i\,b\,j}$ to $b=0$ case
by means of $\langle 5 5 \rangle$:
\begin{equation}\label{eq55}
\mathscr{K}^{cmdn}_{a\,i\,b\,j}= 
q^{-i-j-1} \Bigl(-\mathscr{K}^{c,m-1,d,n-1}_{a,i,b-1,j}
+ \mathscr{K}^{c,m,d,n}_{a,i+1,b-1,j+1}
+q^{m+n+1}\mathscr{K}^{c,m,d-1,n}_{a,i,b-1,j}\Bigr),
\end{equation}
which is valid for $b\ge 1$ and 
$a,c,d,m,n,i,j \ge 0$.
This can be fitted to a recursion relation of $q^2$-trinomial coefficients.
The solution reads
\begin{equation}\label{m1}
\begin{split}
\mathscr{K}^{cmdn}_{a\,i\,b\,j}
= &\delta_{c+m+d,a+i+b}\delta_{d+n-c,b+j-a}\sum_{\alpha, \beta}
\left\{{b \atop \alpha, \beta, b-\alpha-\beta}\right\}
(-1)^{\alpha}\\
&\times 
q^{(\alpha+\beta-b)(\alpha+\beta+b-1)+(m+n-2\alpha+1)\beta
-(i+j+1)b}
\mathscr{K}^{c, m-\alpha, d-\beta, n-\alpha}_{
a,i+b-\alpha-\beta,0,j+b-\alpha-\beta},
\end{split}
\end{equation}
where the sum is over $\alpha, \beta \in \Z_{\ge 0}$, which is actually finite.

Second we reduce $d$ by means of $\langle 2 2 \rangle|_{b=0}$:
\begin{align*}
\mathscr{K}^{cmdn}_{a\,i\,0\,j}&= 
\frac{q^{-m-n-1}}{1-q^{4d}} \Bigl(
-(1-q^{2i})(1-q^{2j})\mathscr{K}^{c,m,d-1,n}_{a,i-1,0,j-1}
+(1-q^{2m+2})(1-q^{2n+2})\mathscr{K}^{c,m+1,d-1,n+1}_{a,i,0,j}\Bigr),
\end{align*}
which is valid for $d\ge 1$.
By considering the combination
$\frac{(q^4)_d(q^2)_m(q^2)_n}{(q^2)_i(q^2)_j}
\mathscr{K}^{cmdn}_{a\,i\,0\,j}$,
this is fitted with a recursion relation of $q^2$-binomial coefficients.
The solution reads
\begin{equation}\label{m2}
\begin{split}
\mathscr{K}^{cmdn}_{a\,i\,0\,j}
= \delta_{c+m+d,a+i}\delta_{d+n-c,j-a}& \frac{q^{-(m+n+1)d}}{(q^4)_d}
\left\{{i,j \atop m,n}\right\}
\sum_{\gamma}
\left\{{d,m+d-\gamma, n+d-\gamma \atop 
\gamma, d-\gamma, i-\gamma, j-\gamma}\right\}
(-1)^{\gamma}\\
&\times q^{(\gamma-d)(\gamma+d-1)}
\mathscr{K}^{c,m+d-\gamma,0,n+d-\gamma}_{a, i-\gamma,0,j-\gamma},
\end{split}
\end{equation}
where the sum over $\gamma \in \Z_{\ge 0}$ is finite.
Combining (\ref{m1}) and (\ref{m2}), we obtain 
\begin{align}
&\mathscr{K}^{c m d n}_{a \,i \,b \,j}=
\delta_{c+m+d,a+i+b}\delta_{d+n-c,b+j-a}
\sum_{\alpha,\beta,\gamma}
\frac{(-1)^{\alpha+\gamma}}{(q^4)_{d-\beta}}q^{\phi_1}
\mathscr{K}^{c, m+d-\alpha-\beta-\gamma,0, 
n+d-\alpha-\beta-\gamma}_{a, i+b-\alpha-\beta-\gamma ,
0, j+b-\alpha-\beta-\gamma}\nonumber\\
&\times \left\{{b, d-\beta,  i+b-\alpha-\beta, j+b-\alpha-\beta, 
m+d-\alpha-\beta-\gamma, n+d-\alpha-\beta-\gamma\atop
\alpha, \beta, \gamma, m-\alpha, n-\alpha, b-\alpha-\beta, d-\beta-\gamma,
i+b-\alpha-\beta-\gamma, j+b-\alpha-\beta-\gamma}\right\},
\label{bd2}
\end{align}
where $\phi_1$ is given in (\ref{bd}).
 
Next we reduce $n$ and $j$  in $\mathscr{K}^{c m 0 n}_{a\,i\,0\,j} $ to $0$
keeping $d=b=0$.
Such recursion relations are available 
from $\langle 2 4\rangle|_{b=d=0}$ and 
$\langle 3 5\rangle|_{b=d=0}$:
\begin{align}
\mathscr{K}^{c m 0 n}_{a\,i\,0\,j} &=\frac{1}{1-q^{2n}}\left(
q^{2 a+i-m}(1-q^{2j})\mathscr{K}^{c,m,0,n-1}_{a,i,0,j-1}
+
q^{j-m}(1-q^{2i})\mathscr{K}^{c,m,0,n-1}_{a+1,i-1,0,j}\right),
\label{recnj1}\\
\mathscr{K}^{c m 0 n}_{a\,i\,0\,j} &=
q^{2c-i+m}\mathscr{K}^{c,m,0,n-1}_{a,i,0,j-1}
+q^{n-i}(1-q^{4c+4})\mathscr{K}^{c+1,m-1,0,n}_{a,i,0,j-1},
\label{recnj2}
\end{align}
which hold for $n\ge 1$ and $j\ge 1$, respectively.
Note that either (\ref{recnj1})$|_{j=0}$   
or (\ref{recnj2})$|_{n=0}$ leads to 
\begin{equation*}
\mathscr{K}^{a i 0 0}_{a i 0 0} = (-1)^iq^{2(a+1)i}
\end{equation*}
with the help of the conservation law (\ref{claw}) and 
$\mathscr{K}^{0 0 0 0}_{0 0 0 0}=1$.
It is easy to solve (\ref{recnj1}) and (\ref{recnj2}) 
with the above initial condition.
The solution is given by (\ref{rest}).
It fulfills the symmetry 
\begin{equation}\label{chie}
\mathscr{K}^{c m 0 n}_{a\, i\, 0\, j} = \frac{(q^4)_a}{(q^4)_c}
\left\{{ i, j \atop m, n} \right\}
\mathscr{K}^{a\,i\,0\, j}_{c m 0 n}
\end{equation}
in accordance with (\ref{Ksym}). 
This is seen by replacing $\lambda$ with $a-c+\lambda$
in (\ref{rest}).
Finally (\ref{bd}) is obtained by applying (\ref{chie}) 
to (\ref{bd2}).

\section{Outline of the proof of Theorem \ref{th:combk}}\label{app:kq0}

Let $\langle rs\rangle$ be the equation for 
$\mathscr{K}^{cmdn}_{a\,i\,b\,j}$
as in Appendix \ref{app:K}.
We are going to prove Theorem \ref{th:combk} solely by using the 
$\langle rs\rangle$'s,  i.e. the characterization of $\mathscr{K}$  
without relying on the 
explicit formula in Theorem \ref{th:main}.

\smallskip
{\em Proof of claim (i) in Theorem \ref{th:combk}}.
From Theorem \ref{th:main}, 
$\mathscr{K}^{cmdn}_{a\, i\, b\, j} $ is a rational function of $q$.
We divide the proof into 2 Steps.

{\em Step 1}. We show $\mathscr{K}^{cmdn}_{a\, i\, b\, j} \in \Z[q,q^{-1}]$.
First we see that (\ref{eq55})
attributes the claim to $b=0$ case $\mathscr{K}^{c m d n}_{a\, i\, 0\, j} $.
One can utilize similar relations $\langle 35\rangle |_{b=0}$, 
$\langle 45\rangle |_{b=j=0}$ and 
$\langle 44\rangle |_{b=j=i=0}$ with coefficients in $\Z[q,q^{-1}]$ 
to reduce the indices $j, i,a$ successively and thereby the claim 
itself to 
$\mathscr{K}^{c m d n}_{0\, 0\, 0\, 0}$.
But the last quantity is 
$\delta_{c 0}\delta_{m 0}\delta_{d 0}\delta_{n 0}$
by the conservation law and we are done.

\smallskip
{\em Step 2}. We show $\mathscr{K}^{cmdn}_{a\, i\, b\, j} \in A$, 
where $A$ is the ring of rational functions of $q$ regular at $q=0$.
(Plainly, $\lim_{q\rightarrow 0}\mathscr{K}^{cmdn}_{a\, i\, b\, j} <\infty$.)
Our strategy is similar to Step 1 but this time 
with coefficients from $A$ instead of  
$\Z[q,q^{-1}]$. 
First solving  (\ref{eq55}) for $\mathscr{K}^{c,m,d,n}_{a,i+1,b-1, j+1} $ we see that 
the claim $\mathscr{K}^{cmdn}_{a\, i\, b\, j} \in A$ for $i,j \ge 1$ 
is attributed to $\min(i,j)=0$ case.
Thus we consider the reductions of 
$\mathscr{K}^{cmdn}_{a\, i\, b\, 0}$ and 
$\mathscr{K}^{cmdn}_{a\, 0\, b\, j} $. 
For the former, 
$\langle45\rangle|_{j=0}$ provides the relation
$\mathscr{K}^{cmdn}_{a\, i\, b\, 0} = 
\sum A \mathscr{K}^{\bullet, \bullet,\bullet, \bullet}_{\bullet, i-1, \bullet, 0}$.
For the latter one eliminates $\mathscr{K}^{cmdn}_{a\, 1\, b\, j}$ between 
$\langle54\rangle|_{i=0}$ and $\langle 55\rangle|_{i=0}$
to get a relation
$\mathscr{K}^{cmdn}_{a\, 0\, b\, j}  = 
\sum A  \mathscr{K}^{\bullet, \bullet,\bullet, \bullet}_{\bullet, 0, \bullet, j-1}
+ \sum A  \mathscr{K}^{\bullet, \bullet,\bullet, \bullet}_{\bullet, 0, \bullet, j-2}$.
Thus the claim is reduced to $\mathscr{K}^{cmdn}_{a\, 0\, b\, 0}$
in the both cases.
Next we utilize $\langle 44 \rangle|_{i=j=0}$ having the form
$\mathscr{K}^{cmdn}_{a\, 0\, b\, 0} = 
\sum A \mathscr{K}^{\bullet, m-1,\bullet, \bullet}_{\bullet, 0, \bullet, 0}
+\sum A \mathscr{K}^{\bullet, m-2,\bullet, \bullet}_{\bullet, 0, \bullet, 0}$.
Thus the claim is reduced to $\mathscr{K}^{c\, 0\, d\,n}_{a\, 0\, b\, 0}$. 
Now we eliminate $\mathscr{K}^{c\, 1\, d\,n}_{a\, 0\, b\, 0}$ between
$\langle22\rangle|_{i=j=m=0}$ and $\langle 44\rangle|_{i=j=m=0}$
to get 
$\mathscr{K}^{c\, 0\, d\, n}_{a\, 0\, b\, 0} = 
\sum A \mathscr{K}^{\bullet, 0,\bullet, \bullet}_{a-1, 0, \bullet, 0}$
reducing the claim further to $\mathscr{K}^{c\, 0\, d\,n}_{0\, 0\, b\, 0}$. 
Similarly eliminate $\mathscr{K}^{c\, 1\, d\,n}_{0\, 0\, b\, 0}$ between
$\langle22\rangle|_{i=j=m=a=0}$ and $\langle 23\rangle|_{i=j=m=a=0}$
to get 
$\mathscr{K}^{c\, 0\, d\, n}_{0\, 0\, b\, 0} = 
\sum A \mathscr{K}^{\bullet, 0, d-1, \bullet}_{0, 0, \bullet, 0}$
reducing the claim down to 
$\mathscr{K}^{c\, 0\, 0\,n}_{0\, 0\, b\, 0} 
= \delta_{c,b}\delta_{n,2b}\mathscr{K}^{b\, 0\, 0\,2b}_{0\, 0\, b\, 0} $,
where the last equality is due to the conservation law (\ref{claw}).
Finally we eliminate 
$\mathscr{K}^{c\, 0\, 0\,n}_{0\, 1\, b\, 0}$ 
between
$\langle24\rangle|_{i=j=m=a=d=0}$ and $\langle 45\rangle|_{i=j=m=a=d=0}$
to find
$\mathscr{K}^{b\, 0\, 0\, 2b}_{0\, 0\, b\, 0} = (b\rightarrow b-1)
= \cdots = \mathscr{K}^{0000}_{0000}=1$.
The claim (i) has been proved.

\smallskip
{\em Proof of claim (ii) in Theorem \ref{th:combk}}.
By the definition (\ref{ckdef}) of $\mathcal{K}$ and (\ref{Ksym}) 
we a priori know 
$\mathcal{K}^{cmdn}_{a\, i\, b \, j} = \mathcal{K}^{a\, i\, b \, j}_{cmdn}$.
The claim (ii) must be consistent with this property.
It postulates that the map 
$(c,m,d,n)\mapsto (c',m',d',n')$ defined by 
the piecewise linear formula (\ref{combk}) must be involutive.
It is certainly so because the formula arises as the ultradiscretization of 
${\bf K}$ and ${\bf K} = {\bf K}^{-1}$ holds.
See the comment after (\ref{bKdef}).

By the argument so far, 
the formula (\ref{combk}) for 
$\mathcal{K}^{cmdn}_{a\, i\, b\, j}$ is equivalent to that 
obtained under the interchange $(c,m,d,n) \leftrightarrow (a,i,b,j)$.
This observation halves our task in what follows.
From  (\ref{claw})  we will always take it for granted 
that $\mathcal{K}^{cmdn}_{a\, i\, b\, j}$ also 
obeys the same conservation law.

Now we begin reducing the indices as in the proof of the claim (i)
by using $\langle r s \rangle|_{q=0}$.
Consider for example (\ref{eq55}), multiply it with 
$q^{i+j+1}$ and set $q=0$.
Now that we know $\mathscr{K}$'s are polynomials in $q$,
the result gives 
$\mathcal{K}^{c,m,d,n}_{a,i+1,b-1, j+1} = 
\mathcal{K}^{c,m-1,d,n-1}_{a,\, i,\, b-1,\, j}$ for the constant terms.
Similarly $\langle 33\rangle|_{q=0}$ yields 
$\mathcal{K}^{c,m,d,n}_{a-1,i+1,b-1,j+1} = 
\mathcal{K}^{c+1,m-1,d+1,n-1}_{a,\, i,\, b,\, j}$.
One can check that the piecewise linear formula (\ref{combk}) also 
satisfies these recursion relations.
Thus a proof for 
$\mathcal{K}^{cmdn}_{a\, i\, b \, j}$ is attributed to 
the situation $\min(a,b,c,d)=\min(m,n,i,j)=0$.
There are $4\times 4=16$ such possibilities but 
the symmetry  $(a,m,d,n) \leftrightarrow (a,b,i,j)$
already established in the above halves them to the 8 cases
as follows. (Their choice is not unique.)
\begin{equation*}
\begin{split}
& \mathrm{(I)}\; d=n=0, \quad\;\;\;\mathrm{(II)}\;  c=n=0, 
\quad\;\;\,\mathrm{(III)}\; b=n=0, \quad\;\;\,\;\mathrm{(IV)}\; a=n=0,\\
&\mathrm{(V)}\; d=m=0, \quad\mathrm{(VI)}\;c=m=0, 
\quad\mathrm{(VII)}\; b=m=0, \quad\mathrm{(VIII)}\; a=m=0.
\end{split}
\end{equation*}
We have proved them all case by case. 
Here we illustrate the proof for (I) only.
The other cases can be treated similarly.

For (I), the formula (\ref{combk}) to be proved reads
$\mathcal{K}^{c m 0 0}_{a\,i\,b\,j} = 
\delta_{a,c+m}\delta_{i,0}\delta_{b,0}\delta_{j,m}$. 
First note that (\ref{eq55})$|_{d=n=0}$ in the limit $q\rightarrow 0$
tells that $\mathcal{K}^{c m 0 0}_{a\, i\, b\, j} = 0$ if 
$\min(i,j)\ge 1$.
This agrees with the formula to be shown in the region $\min(i,j)\ge 1$.
Thus we are left to verify  
(I-1) $\mathcal{K}^{c m 0 0}_{a\, 0\, b\, j} = 
\delta_{a,c+m}\delta_{b,0}\delta_{j,m}$ and 
(I-2) $\mathcal{K}^{c m 0 0}_{a\,i\,b\, 0} = 
\delta_{a,c}\delta_{i,0}\delta_{b,0}\delta_{0,m}$.

To prove (I-1), we utilize
$\langle 54 \rangle |_{d=n=i=0}$  in the limit $q\rightarrow 0$, which says
$\mathcal{K}^{c m 0 0}_{a\,0\,b\,j} = \mathcal{K}^{c,m-1,0,0}_{a-1,0,b,j-1}$.
Note that the RHS of (I-1) is also invariant under the 
simultaneous decrement of $a,m,j$ by $1$.
Since the conservation law implies $a,m \ge j$ for nonzero 
$\mathcal{K}^{c m 0 0}_{a\,0\,b\,j}$, the above recursion reduces (I-1) to 
$j=0$ case, namely 
$\mathcal{K}^{c m 0 0}_{a\,0\,b\,0}=\delta_{a,c}\delta_{b,0}\delta_{m,0}$.
By the conservation law, the LHS equals 
$\delta_{a, b+c}\delta_{m,2b}\mathcal{K}^{c,2b,0,0}_{b+c,0,b,0}$.
Thus we are to show 
$\mathcal{K}^{c,2b,0,0}_{b+c,0,b,0} = \delta_{b,0}$.
From $\langle 42 \rangle|_{d=n=i=j=0}$  in the limit $q\rightarrow 0$, we find that 
$\mathcal{K}^{c,2b,0,0}_{b+c,0,b,0}$ indeed vanishes for $b\ge 1$.
It remains to check $\mathcal{K}^{c 0 0 0}_{c 0 0 0}=1$.
From $\langle 4 4 \rangle|_{d=n=i=j=m=b=0}$ (with $q$ generic),
we find $\mathscr{K}^{c 0 0 0}_{c 0 0 0}=(c\rightarrow c-1) = \cdots = 
\mathscr{K}^{0 0 0 0}_{0 0 0 0}=1$, completing the proof of (I-1).

To prove (I-2), we utilize
$\langle 23 \rangle |_{d=n=j=0}$ in the limit $q\rightarrow 0$, which says
$\mathcal{K}^{c m 0 0}_{a\,i\,b\,0}
= \mathcal{K}^{c,m-1,0,0}_{a-1,i+1,b-1,0}$.
Since the conservation law implies $a,m \ge b$ for nonzero 
$\mathcal{K}^{c m 0 0}_{a\,i\,b\,0}$, 
the above recursion reduces the LHS of (I-2) to 
$b=0$ case.  Thus we are to verify  
$\mathcal{K}^{c,m-b,0,0}_{a-b,i+b,0,0}
=\delta_{a,c}\delta_{i,0}\delta_{b,0}\delta_{0,m}$.
Since the conservation law limits the nontrivial case to $a=b+c$ and $m=i+2b$, 
it boils down to showing $\mathcal{K}^{\alpha \beta 0 0}_{\alpha \beta 0 0} 
= \delta_{\beta,0}$.
The vanishing for $\beta \ge 1$ follows from 
$\langle 23 \rangle |_{d=n=j=b=0}$ in the limit $q\rightarrow 0$,
and $\beta=0$ case has already been shown 
in the end of the proof of (I-1).
We have finished the proof of (I-2) and thereby (I).

\section*{Acknowledgments}
A part of this work was presented in 
The XXIX International Colloquium on Group-Theoretical Methods in Physics,
August 20-26, 2012 at Chern Institute of Mathematics, Tianjin, China.
The authors thank Alexey Isaev for a generous explanation of the
tetrahedron reflection equation and an important remark.
They also thank Murray Batchelor, Vladimir Bazhanov,  Anatol Kirillov, 
Igor Korepanov and Peter Kulish for kind interest and comments.
Thanks are also due to 
Toshiki Nakashima, Soichi Okada, Itaru Terada and Hiroyuki Yamane 
for communications.
A.K. thanks Vladimir Mangazeev for kind interest in the previous work 
and Sergey Sergeev for showing 
an unpublished review article, which 
provided useful information to initiate this work.
This work is supported by Grants-in-Aid for
Scientific Research No.~23340007, No.~24540203 and 
No.~23654007 from JSPS.

\end{document}